\documentclass[11pt,twoside,a4paper,notitlepage]{article}

\usepackage{ifthen}

\newboolean{conferenceversion}
\setboolean{conferenceversion}{false}

\usepackage[a4paper,margin=2.5cm]{geometry}
\usepackage{times}

\usepackage[english]{babel}

\usepackage{fancyhdr}

\usepackage[sf,SF]{subfigure}

\usepackage{bussproofs}

\usepackage{tikz}
\usepackage{tikz-3dplot}
\usetikzlibrary{shapes}

\makeatletter{}%
\usepackage{ifthen}

\makeatletter{}%

\usepackage{ifthen}

\newboolean{detectedSTOC}
\newboolean{detectedFOCS}
\newboolean{detectedElsevier}
\newboolean{detectedNOW}
\newboolean{detectedLMCS}
\newboolean{detectedPoster}
\newboolean{detectedSIAM}
\newboolean{detectedIEEE}
\newboolean{detectedLNCS}
\newboolean{detectedACM}
\newboolean{detectedSigplanconf}
\newboolean{detectedToC}
\newboolean{detectedLIPIcs}
\newboolean{detectedAAAI}
\newboolean{detectedCompCplx}
\newboolean{detectedArticle}
\newboolean{detectedReport}
\newboolean{detectedThesis}

\makeatletter

\@ifclassloaded{sig-alternate}
{\setboolean{detectedSTOC}{true}}
{\setboolean{detectedSTOC}{false}}

\@ifclassloaded{elsarticle}
{\setboolean{detectedElsevier}{true}}
{\setboolean{detectedElsevier}{false}}

\@ifclassloaded{now}
{\setboolean{detectedNOW}{true}}
{\setboolean{detectedNOW}{false}}

\@ifclassloaded{lmcs}
{\setboolean{detectedLMCS}{true}}
{\setboolean{detectedLMCS}{false}}

\@ifclassloaded{IEEEtran} {
  \setboolean{detectedIEEE}{true}
  \ifCLASSOPTIONconference {
    \setboolean{detectedFOCS}{true}
  }
  \else {
    \setboolean{detectedFOCS}{false}
  }
  \fi
}
{
  \setboolean{detectedFOCS}{false}
  \setboolean{detectedIEEE}{false}
}

\@ifclassloaded{siamltex1213} 
{\setboolean{detectedSIAM}{true}}
{\setboolean{detectedSIAM}{false}}

\@ifclassloaded{llncs}
{\setboolean{detectedLNCS}{true}}
{\setboolean{detectedLNCS}{false}}

\@ifclassloaded{acmsmall}
{\setboolean{detectedACM}{true}}
{\setboolean{detectedACM}{false}}

\@ifclassloaded{sigplanconf}
{\setboolean{detectedSigplanconf}{true}}
{\setboolean{detectedSigplanconf}{false}}

\@ifclassloaded{toc}
{\setboolean{detectedToC}{true}}
{\setboolean{detectedToC}{false}}

\@ifclassloaded{lipics}
{\setboolean{detectedLIPIcs}{true}}
{\@ifclassloaded{lipics-v2016}
  {\setboolean{detectedLIPIcs}{true}}
  {\setboolean{detectedLIPIcs}{false}}
}

\@ifclassloaded{cc}
{\setboolean{detectedCompCplx}{true}}
{\setboolean{detectedCompCplx}{false}}

\@ifpackageloaded{aaai}
{\setboolean{detectedAAAI}{true}}
{\setboolean{detectedAAAI}{false}}

\@ifclassloaded{sciposter}
{\setboolean{detectedPoster}{true}}
{\setboolean{detectedPoster}{false}}

\@ifclassloaded{article}
{\setboolean{detectedArticle}{true}}
{\setboolean{detectedArticle}{false}}

\makeatother

\ifthenelse{\not \isundefined{\examen} 
  \and \not \isundefined{\disputationsdatum} 
  \and \not \isundefined{\disputationslokal}}   
  {\setboolean{detectedThesis}{true}}
  {\setboolean{detectedThesis}{false}}

\ifthenelse{\boolean{detectedArticle} \or \boolean{detectedThesis}
  \or \boolean{detectedSTOC} \or \boolean{detectedFOCS}
  \or \boolean{detectedSIAM} \or \boolean{detectedIEEE}
  \or \boolean{detectedPoster}}
{\setboolean{detectedReport}{false}}
{\setboolean{detectedReport}{true}}

\ifthenelse{\boolean{detectedLIPIcs}}
{}
{\usepackage[utf8]{inputenc}}

\ifthenelse{\boolean{detectedIEEE}}
{\usepackage[cmex10]{amsmath}
 \interdisplaylinepenalty=2500}
{\usepackage{amsmath}}
\usepackage{amssymb}
\usepackage{amsfonts}

\usepackage{mathtools}

\ifthenelse
{\boolean{detectedElsevier}}
{}
{\usepackage{varioref}}

\usepackage{xspace}

\ifthenelse    
{\boolean{detectedSTOC}      \or \boolean{detectedSIAM}         \or 
  \boolean{detectedLMCS}     \or \boolean{detectedNOW}          \or 
  \boolean{detectedLNCS}     \or \boolean{detectedACM}          \or
  \boolean{detectedToC}      \or \boolean{detectedLIPIcs}       \or
  \boolean{detectedCompCplx} \or \boolean{detectedSigplanconf}}
{}
{
  \usepackage{amsthm}
  \usepackage{amsthm-boldfont}
}

\ifthenelse        
{\boolean{detectedElsevier}  \or \boolean{detectedFOCS}         \or 
  \boolean{detectedSTOC}     \or \boolean{detectedPoster}       \or
  \boolean{detectedSIAM}     \or \boolean{detectedLMCS}         \or
  \boolean{detectedIEEE}     \or \boolean{detectedNOW}          \or
  \boolean{detectedToC}      \or \boolean{detectedThesis}       \or
  \boolean{detectedLNCS}     \or \boolean{detectedACM}          \or 
  \boolean{detectedLIPIcs}   \or \boolean{detectedAAAI}         \or
  \boolean{detectedCompCplx} \or \boolean{detectedSigplanconf}}
{}
{\usepackage{jncaptionsheadings}}

\makeatletter{}%

\usepackage{ifthen}
\usepackage{xspace}
\ifthenelse
{\boolean{detectedElsevier}}
{}
{\usepackage{varioref}}

\DeclareMathAlphabet{\mathsfsl}{OT1}{cmss}{m}{sl}

\newcommand{\formatfunctiontoset}[1]{\mathit{#1}}

\newcommand{\introduceterm}[1]{{\emph{#1}}}

\newcommand{\eqperiod}{\enspace .}
\newcommand{\eqcomma}{\enspace ,}

\newcommand{\wrt}{with respect to\xspace}

\ifthenelse{\boolean{detectedToC}}{}
  { %
    }

\ifthenelse{\boolean{detectedToC}}{}{\newcommand{\ie}{i.e.,\ }}

\ifthenelse{\boolean{detectedLIPIcs}}
{}
{}

\newcommand{\etal}{et al.\@\xspace}

\ifthenelse{\boolean{detectedIEEE}}{}{}

\newcommand{\bigoh}[1]{\mathrm{O} ( #1 )}

\newcommand{\littleoh}[1]{\mathrm{o} ( #1 )}

\newcommand{\bigtheta}[1]{\Theta ( #1 )}
\newcommand{\BIGOMEGA}[1]{\Omega \left( #1 \right)}
\newcommand{\Bigomega}[1]{\Omega \bigl( #1 \bigr)}
\newcommand{\bigomega}[1]{\Omega ( #1 )}

\newcommand{\complclassformat}[1]%
        {\textrm{\upshape{\textsf{#1}}}\xspace}

\newcommand{\cocomplclass}[1]%
        {\textrm{\upshape{\textsf{co#1}}}\xspace}

\newcommand{\DTIMEadviceclass}[2]%
    {\ensuremath{\complclassformat{DTIME}\bigl(#1\bigr)/{#2}}}
\newcommand{\Pclass}{\complclassformat{P}}
\newcommand{\NP}{\complclassformat{NP}}

\newcommand{\coNP}{\cocomplclass{NP}}

\newcommand{\PSPACE}{\complclassformat{PSPACE}}

\newcommand{\EXPTIME}{\complclassformat{EXPTIME}}

\newcommand{\refsec}[1]{Section~\ref{#1}}

\newcommand{\refapp}[1]{Appendix~\ref{#1}}

\newcommand{\refth}[1]{Theorem~\ref{#1}}

\newcommand{\refthm}[1]{Theorem~\ref{#1}}

\newcommand{\reflem}[1]{Lemma~\ref{#1}}
\newcommand{\reftwolems}[2]{Lemmas~\ref{#1} and~\ref{#2}}

\newcommand{\refdef}[1]{Definition~\ref{#1}}

\newcommand{\refobs}[1]{Observation~\ref{#1}}

\ifthenelse
{\isundefined{\refeq}}
{\newcommand{\refeq}[1]{\eqref{#1}}}
{\renewcommand{\refeq}[1]{\eqref{#1}}}

\newcommand{\Floor}[1]{\bigl \lfloor #1 \bigr \rfloor}

\ifthenelse{\boolean{detectedToC}}{}
{

  \newcommand{\Nplus}     {\mathbb{N}^{+}}

}

\newcommand{\MAXOFEXPR}[2][]{\max_{#1} \left\{ #2 \right\}}
\newcommand{\MINOFEXPR}[2][]{\min_{#1} \left\{ #2 \right\}}
\newcommand{\Maxofexpr}[2][]{\max_{#1} \bigl\{ #2 \bigr\}}
\newcommand{\Minofexpr}[2][]{\min_{#1} \bigl\{ #2 \bigr\}}
\newcommand{\maxofexpr}[2][]{\max_{#1} \{ #2 \}}
\newcommand{\minofexpr}[2][]{\min_{#1} \{ #2 \}}

\newcommand{\MAXOFSET}[3][:]{\max \left\{ #2 #1 #3 \right\}}
\newcommand{\MINOFSET}[3][:]{\min \left\{ #2 #1 #3 \right\}}
\newcommand{\Maxofset}[3][:]{\max \bigl\{ #2 #1 #3 \bigr\}}
\newcommand{\Minofset}[3][:]{\min \bigl\{ #2 #1 #3 \bigr\}}

\renewcommand{\MAXOFSET}[3][:]%
     {\ifthenelse{\equal{#1}{;}}%
     {\MAXOFEXPR{ #2 \,;\, #3 }}
     {\ifthenelse{\equal{#1}{:}}%
     {\MAXOFEXPR{ #2 \,:\, #3 }}
     {\max \twincommandJN{\left\{}{#2}{\left#1}{\right}{\,#3}{\right\}}}}}
\renewcommand{\MINOFSET}[3][:]%
     {\ifthenelse{\equal{#1}{;}}%
     {\MINOFEXPR{ #2 \,;\, #3 }}
     {\ifthenelse{\equal{#1}{:}}%
     {\MINOFEXPR{ #2 \,:\, #3 }}
     {\min \twincommandJN{\left\{}{#2}{\left#1}{\right}{\,#3}{\right\}}}}}

\renewcommand{\Maxofset}[3][:]%
     {\ifthenelse{\equal{#1}{;}}%
     {\Maxofexpr{ #2 \,;\, #3 }}
     {\ifthenelse{\equal{#1}{:}}%
     {\Maxofexpr{ #2 \,:\, #3 }}
     {\max \twincommandJN{\bigl\{}{#2}{\bigl#1}{\bigr}{\,#3}{\bigr\}}}}}
\renewcommand{\Minofset}[3][:]%
     {\ifthenelse{\equal{#1}{;}}%
     {\Minofexpr{ #2 \,;\, #3 }}
     {\ifthenelse{\equal{#1}{:}}%
     {\Minofexpr{ #2 \,:\, #3 }}
     {\min \twincommandJN{\bigl\{}{#2}{\bigl#1}{\bigr}{\,#3}{\bigr\}}}}}

\DeclareMathOperator{\Expop}{E}

\ifthenelse{\boolean{detectedLMCS}}
{}
{}

\newcommand{\twincommandJN}[6]%
    {#1#2#3\vphantom{#2#5}\mspace{-2.05mu}#4.#5#6}

\newcommand{\CondExp}[2]%
    {\Expop\twincommandJN{\bigl[}{#1}{\bigl|}{\bigr}{\,#2}{\bigr]}}
\newcommand{\CONDEXP}[2]%
     {\Expop\twincommandJN{\left[}{#1}{\left|}{\right}{\,#2}{\right]}}

\newcommand{\CondProb}[3][]%
    {\Pr_{#1}\twincommandJN{\bigl[}{#2}{\bigl|}{\bigr}{\,#3}{\bigr]}}
\newcommand{\CONDPROB}[3][]%
    {\Pr_{#1}\twincommandJN{\left[}{#2}{\left|}{\right}{\,#3}{\right]}}

\newcommand{\isdistras}[2]{\ensuremath{#1} \sim \ensuremath{#2}}

\newcommand{\descname}{\formatfunctiontoset{desc}}

\ifthenelse{\isundefined{\descnode}}
{\newcommand{\descnode}[2][]{\descname_{#1}(#2)}}
{\renewcommand{\descnode}[2][]{\descname_{#1}(#2)}}

\newcommand{\setcompact}[1]{{\ensuremath{\bigl\{ #1 \bigr\}}}}

\newcommand{\setdescrcompact}[3][\mid]{{\setcompact{ #2 #1 #3 }}}

\newcommand{\setsizesmall}[1]{\ensuremath{\lvert#1\rvert}}

\newcommand{\set}[1]{\{ #1 \}}
\newcommand{\Set}[1]{\bigl\{ #1 \bigr\}}

\newcommand{\setdescr}[3][\mid]{\set{ #2 #1 #3 }}
\newcommand{\Setdescr}[3][|]%
     {\ifthenelse{\equal{#1}{;}}%
     {\Set{ #2 \,;\, #3 }}
     {\ifthenelse{\equal{#1}{:}}%
     {\Set{ #2 \,:\, #3 }}
     {\twincommandJN{\bigl\{}{#2\,}{\bigl#1}{\bigr}{\,#3}{\bigr\}}}}}
\newcommand{\SETDESCR}[3][|]%
     {\twincommandJN{\left\{}{#2\,}{\left#1}{\right}{\,#3}{\right\}}}

\newcommand{\Setdescrbrackets}[3][|]%
     {\twincommandJN{\bigl[}{#2}{\bigl#1}{\bigr}{\,#3}{\bigr]}}
\newcommand{\SETDESCRBRACKETS}[3][|]%
     {\twincommandJN{\left[}{#2}{\left#1}{\right}{\,#3}{\right]}}

\newcommand{\Setsize}[1]{\bigl\lvert#1\bigr\rvert}
\newcommand{\setsize}[1]{\lvert#1\rvert}

\newcommand{\intersection}{\cap}

\newcommand{\union}{\cup}
\newcommand{\Union}{\bigcup}

\newcommand{\unionSP}{\, \union \, }

\newcommand{\disjointunion}{\overset{.}{\cup}}

\newcommand{\DisjointunionInText}%
    {{\smash{\overset{\mbox{\boldmath{.}}}{\bigcup}}}\vphantom{\bigcup}}
\newcommand{\intnfirst}[1]{[{#1}]}

\newcommand{\Lor}{\bigvee}
\newcommand{\Land}{\bigwedge}

\newcommand{\olnot}[1]{\overline{#1}}
\newcommand{\stdnot}[1]{\olnot{#1}}

\newcommand{\cnfform}{\cnfshort for\-mu\-la\xspace}

\newcommand{\cnfshort}{CNF\xspace}

\newcommand{\xcnf}[1]{\mbox{\ensuremath{#1}-CNF}\xspace}

\newcommand{\xcnfform}[1]{\mbox{\ensuremath{#1}-}\cnfform}
\newcommand{\kcnfform}{\xcnfform{\clwidth}}
\newcommand{\xclause}[1]{\mbox{\ensuremath{#1}-clause}\xspace}

\newcommand{\nvar}{n}
\newcommand{\nvars}{\nvar}
\newcommand{\nclause}{m}

\newcommand{\clwidth}{k}

\newcommand{\randkcnfnclwrepl}[3][\clwidth]%
        {\ensuremath{\mathcal{F}^{#2, #3}_{#1}}}
\newcommand{\randkcnfnclwreplstd}%
        {\randkcnfnclwrepl{\clwidth}{\nvar}{\nclause}}

\newcommand{\israndkcnfnclwrepl}[4]%
  {\isdistras{#1}{\randkcnfnclwrepl[#2]{#3}{#4}}}

\newcommand{\randkcnfprobcl}[3]%
        {\ensuremath{\mathcal{F}^{#2}_{#1} \bigl(#3 \bigr)}}

\newcommand{\pcfor}[4][to]{for #2 := #3 #1 #4 do}
\newcommand{\pcformath}[4][to]%
    {\pcfor[#1]{\ensuremath{#2}}{\ensuremath{#3}}{\ensuremath{#4}}}

\newcommand{\pcassigncompact}[2]{#1 := #2}
\newcommand{\pcassignmathcompact}[2]%
        {\pcassigncompact{\ensuremath{#1}}{\ensuremath{#2}}}

\newcommand{\inductionformat}[1]{\textit{#1}}

\newcommand{\BASE}[1][]
        {\inductionformat
                {%
                        \ifthenelse{\equal{#1}{}}%
                                {Base case: }%
                                {Base case (#1):}%
                }%
        }

\ifthenelse{\isundefined{\DONOTLOADHYPERREFPACKAGE}
  \and \not \boolean{detectedSTOC}     \and \not \boolean{detectedFOCS}
  \and \not \boolean{detectedPoster}   \and \not \boolean{detectedElsevier} 
  \and \not \boolean{detectedSIAM}     \and \not \boolean{detectedACM}
  \and \not \boolean{detectedIEEE}     \and \not \boolean{detectedNOW}
  \and \not \boolean{detectedToC}      \and \not \boolean{detectedThesis}
  \and \not \boolean{detectedLNCS}     \and \not \boolean{detectedLIPIcs}
  \and \not \boolean{detectedAAAI}     \and \not \boolean{detectedSigplanconf}
  \and \not \boolean{detectedCompCplx}}
{\usepackage[colorlinks=true,citecolor=blue]{hyperref}}
{}

\ifthenelse{\boolean{detectedSTOC}}                  
  {%
\makeatletter{}%

\newtheorem{theorem}{Theorem}
\newtheorem{lemma}[theorem]{Lemma}
\newtheorem{proposition}[theorem]{Proposition}
\newtheorem{corollary}[theorem]{Corollary}
\newtheorem{observation}[theorem]{Observation}

\newtheorem{definition}[theorem]{Definition}
\newtheorem{claim}[theorem]{Claim}

\newtheorem{conjecture}{Conjecture}
\newtheorem{openproblem}[conjecture]{Open Problem}

\newdef{remark}{Remark}
\newdef{example}{Example}

\newcounter{unnumber}

}
  {\ifthenelse{\boolean{detectedSIAM}}
    {%
\makeatletter{}%

\newtheorem{observation}[theorem]{Observation}
\newtheorem{claim}[theorem]{Claim}

\newtheorem{conjecture}{Conjecture}
\newtheorem{openquestion}{Open Question}

\newtheorem{remarkinner}[theorem]{Remark}
\newtheorem{exampleinner}[theorem]{Example}

\newcommand{\exampleendmarker}{\qquad$\Diamond$}
\newcommand{\remarkendmarker}{\qquad$\Diamond$}

\newenvironment{example}                        
    {\begin{exampleinner} \rm}
    {\exampleendmarker\end{exampleinner}}

\newenvironment{remark}                        
    {\begin{remarkinner} \rm}
    {\remarkendmarker\end{remarkinner}}

\newcounter{unnumber}

}
    {\ifthenelse{\boolean{detectedNOW}}
      {%
\makeatletter{}%
\newtheorem{standardlocalcounter}{Dummy}[chapter]
\newtheorem{standardglobalcounter}{Dummy}

\newtheorem{theorem}[standardlocalcounter]{Theorem}
\newtheorem{lemma}[standardlocalcounter]{Lemma}
\newtheorem{proposition}[standardlocalcounter]{Proposition}
\newtheorem{corollary}[standardlocalcounter]{Corollary}
\newtheorem{observation}[standardlocalcounter]{Observation}
\newtheorem{fact}[standardlocalcounter]{Fact}

\newtheorem{conjecturelocalcounter}[standardlocalcounter]{Conjecture}
\newtheorem{conjectureglobalcounter}[standardglobalcounter]{Conjecture}
\newtheorem{conjecture}[standardglobalcounter]{Conjecture}
\newtheorem{openquestion}[standardglobalcounter]{Open Question}
\newtheorem{openproblem}[standardglobalcounter]{Open Problem}
\newtheorem{problem}{Problem}

\newtheorem{property}[standardlocalcounter]{Property}
\newtheorem{definition}[standardlocalcounter]{Definition}
\newtheorem{claim}[standardlocalcounter]{Claim}

\newtheorem{algorithm}[standardlocalcounter]{Algorithm}

\newtheorem{remark}[standardlocalcounter]{Remark}
\newtheorem{example}[standardlocalcounter]{Example}

\renewenvironment{proof}[1][Proof]{\par\trivlist
   \item[\hskip \labelsep{\itshape {#1}.}]\prooffont}
   {\hspace*{0pt plus1fill}\fboxsep2.5pt\fboxrule.5pt\raise3pt\hbox{\fbox{}}\endtrivlist}
}
      {\ifthenelse{\boolean{detectedLMCS}}
        {%
\makeatletter{}%

\theoremstyle{plain}    

\newtheorem{theorem}[thm]{Theorem}
\newtheorem{lemma}[thm]{Lemma}
\newtheorem{proposition}[thm]{Proposition}
\newtheorem{corollary}[thm]{Corollary}
\newtheorem{observation}[thm]{Observation}

\newtheorem{conjecture}[thm]{Conjecture}
\newtheorem{problem}[thm]{Problem}

\newtheorem{openquestion}{Open Question}
\newtheorem{openproblem}{Open Problem}

\theoremstyle{definition}

\newtheorem{property}[thm]{Property}
\newtheorem{definition}[thm]{Definition}
\newtheorem{claim}[thm]{Claim}
\newtheorem{remark}[thm]{Remark}
\newtheorem{example}[thm]{Example}

}
        {\ifthenelse{\boolean{detectedIEEE}}
          {%
\makeatletter{}%
\newtheorem{standardlocalcounter}{Dummy}[section]
\newtheorem{standardglobalcounter}{Dummy}

\theoremstyle{plain}    

\newtheorem{theorem}[standardglobalcounter]{Theorem}
\newtheorem{lemma}[standardglobalcounter]{Lemma}
\newtheorem{proposition}[standardglobalcounter]{Proposition}
\newtheorem{corollary}[standardglobalcounter]{Corollary}
\newtheorem{observation}[standardglobalcounter]{Observation}
\newtheorem{fact}[standardglobalcounter]{Fact}

\newtheorem{conjecture}[standardglobalcounter]{Conjecture}
\newtheorem{openquestion}{Open Question}
\newtheorem{openproblem}{Open Problem}
\newtheorem{problem}{Problem}

\theoremstyle{definition}

\newtheorem{property}[standardglobalcounter]{Property}
\newtheorem{definition}[standardglobalcounter]{Definition}
\newtheorem{claim}[standardglobalcounter]{Claim}

\theoremstyle{remark}
\newtheorem{remark}[standardglobalcounter]{Remark}
\newtheorem{example}[standardglobalcounter]{Example}

\newtheoremstyle{meta}%
  {3pt}%
  {3pt}%
  {\scshape \small }%
  {}%
  {\scshape \small }%
  {:}%
  { }%
  {}%

\theoremstyle{meta}
\newtheorem{meta}{Meta comment}

\newtheoremstyle{questions}%
  {3pt}%
  {3pt}%
  {\sffamily \slshape}%
  {}%
  {\bfseries \sffamily \slshape}%
  {:}%
  { }%
  {}%

\theoremstyle{questions}
\newtheorem{questions}{Open questions}

}
          {\ifthenelse{\boolean{detectedLNCS}}
            {%
\makeatletter{}%

\spnewtheorem*{proofsketch}{Proof sketch}{\itshape}{\rmfamily}

\spnewtheorem{observation}{Observation}{\bfseries}{\itshape}
\spnewtheorem{fact}{Fact}{\bfseries}{\itshape}

}
            {\ifthenelse{\boolean{detectedACM}}
              {%
\makeatletter{}%

\newtheorem{observation}[theorem]{Observation}
\newtheorem{fact}[theorem]{Fact}
\newtheorem{claim}[theorem]{Claim}
\newtheorem{openquestion}{Open Question}
\newtheorem{openproblem}{Open Problem}

\newcounter{unnumber}

}
              {\ifthenelse{\boolean{detectedToC}}
                {%
\makeatletter{}%

\theoremstyle{plain}
\newtheorem{observation}[theorem]{Observation}
\newtheorem{openproblem}[theorem]{Open Problem}

\theoremstyle{definition}
\newtheorem{property}[theorem]{Property}

\renewcommand{\refth}[1]{\expref{Theorem}{#1}}
\renewcommand{\reflem}[1]{\expref{Lemma}{#1}}

\renewcommand{\refdef}[1]{\expref{Definition}{#1}}

\renewcommand{\refobs}[1]{\expref{Observation}{#1}}

\renewcommand{\refsec}[1]{\expref{Section}{#1}}

\renewcommand{\refapp}[1]{\expref{Appendix}{#1}}

}
                {\ifthenelse{\boolean{detectedLIPIcs}}
                  {%
\makeatletter{}%

\theoremstyle{plain}    

\newtheorem{fact}[theorem]{Fact}
\newtheorem{proposition}[theorem]{Proposition}
\newtheorem{observation}[theorem]{Observation}
\newtheorem{claim}[theorem]{Claim}
}
                  {\ifthenelse{\boolean{detectedCompCplx}}
                    {%
\makeatletter{}%
}
                    {\ifthenelse{\boolean{detectedSigplanconf}}
                      {%
\makeatletter{}%

\usepackage{amsthm}          
\usepackage{amsthm-boldfont}
\usepackage{url}

\newtheorem{standardlocalcounter}{Dummy}[section]
\newtheorem{standardglobalcounter}{Dummy}

\theoremstyle{plain}    

\newtheorem{theorem}[standardlocalcounter]{Theorem}
\newtheorem{lemma}[standardlocalcounter]{Lemma}
\newtheorem{proposition}[standardlocalcounter]{Proposition}
\newtheorem{corollary}[standardlocalcounter]{Corollary}
\newtheorem{observation}[standardlocalcounter]{Observation}
\newtheorem{fact}[standardlocalcounter]{Fact}

\newtheorem{conjecturelocalcounter}[standardlocalcounter]{Conjecture}
\newtheorem{conjectureglobalcounter}[standardglobalcounter]{Conjecture}
\newtheorem{conjecture}[standardglobalcounter]{Conjecture}
\newtheorem{openquestion}[standardglobalcounter]{Open Question}
\newtheorem{openproblem}[standardglobalcounter]{Open Problem}
\newtheorem{problem}[standardglobalcounter]{Problem}
\newtheorem{question}[standardglobalcounter]{Question}

\theoremstyle{definition}

\newtheorem{property}[standardlocalcounter]{Property}
\newtheorem{definition}[standardlocalcounter]{Definition}
\newtheorem{claim}[standardlocalcounter]{Claim}
\newtheorem{subclaim}[standardlocalcounter]{Subclaim}

\newtheorem{algorithm}[standardlocalcounter]{Algorithm}

\theoremstyle{remark}
\newtheorem{remark}[standardlocalcounter]{Remark}
\newtheorem{example}[standardlocalcounter]{Example}

}
                      {\ifthenelse{\boolean{detectedArticle} 
                          \or \boolean{detectedElsevier}}
                        {%
\makeatletter{}%
\newtheorem{standardlocalcounter}{Dummy}[section]
\newtheorem{standardglobalcounter}{Dummy}

\theoremstyle{plain}    

\newtheorem{theorem}[standardlocalcounter]{Theorem}
\newtheorem{lemma}[standardlocalcounter]{Lemma}
\newtheorem{proposition}[standardlocalcounter]{Proposition}
\newtheorem{corollary}[standardlocalcounter]{Corollary}
\newtheorem{observation}[standardlocalcounter]{Observation}

\newtheorem{conjecturelocalcounter}[standardlocalcounter]{Conjecture}
\newtheorem{conjectureglobalcounter}[standardglobalcounter]{Conjecture}
\newtheorem{conjecture}[standardglobalcounter]{Conjecture}
\newtheorem{openquestion}[standardglobalcounter]{Open Question}
\newtheorem{openproblem}[standardglobalcounter]{Open Problem}
\newtheorem{problem}[standardglobalcounter]{Problem}

\theoremstyle{definition}

\newtheorem{property}[standardlocalcounter]{Property}
\newtheorem{definition}[standardlocalcounter]{Definition}
\newtheorem{claim}[standardlocalcounter]{Claim}

\theoremstyle{remark}
\newtheorem{remark}[standardlocalcounter]{Remark}
\newtheorem{example}[standardlocalcounter]{Example}

}
                        {%
\makeatletter{}%
\newtheorem{standardlocalcounter}{Dummy}[chapter]
\newtheorem{standardglobalcounter}{Dummy}

\theoremstyle{plain}    

\newtheorem{theorem}[standardlocalcounter]{Theorem}
\newtheorem{lemma}[standardlocalcounter]{Lemma}

\newtheorem{observation}[standardlocalcounter]{Observation}

\theoremstyle{definition}

\newtheorem{definition}[standardlocalcounter]{Definition}
\newtheorem{claim}[standardlocalcounter]{Claim}

\theoremstyle{remark}

\newtheoremstyle{meta}%
  {3pt}%
  {3pt}%
  {\scshape \small }%
  {}%
  {\scshape \small }%
  {:}%
  { }%
  {}%

\theoremstyle{meta}

\newtheoremstyle{questions}%
  {3pt}%
  {3pt}%
  {\sffamily \slshape}%
  {}%
  {\bfseries \sffamily \slshape}%
  {:}%
  { }%
  {}%

\theoremstyle{questions}

}}}}}}}}}}}}

\ifthenelse{\boolean{detectedThesis}}
{}
{\usepackage{ifpdf}}

\ifthenelse{\boolean{detectedToC}}
{}
{\usepackage{graphicx}}

\ifpdf         
\fi

\ifthenelse
{\boolean{detectedSTOC}  \or \boolean{detectedThesis} \or 
 \boolean{detectedLNCS}  \or \boolean{detectedToC}    \or 
 \boolean{detectedLIPIcs} \or \boolean{detectedAAAI}}
{}
{
  \setcounter{secnumdepth}{3}
  \setcounter{tocdepth}{3}
}

\newcount\hours
\newcount\minutes
\def\SetTime{\hours=\time
\global\divide\hours by 60
\minutes=\hours
\multiply\minutes by 60
\advance\minutes by-\time
\global\multiply\minutes by-1 }
\SetTime
\def\now{\number\hours:\ifnum\minutes<10 0\fi\number\minutes}

\makeatletter{}%
\newcommand{\formuladots}{\cdots}

\newcommand{\boundary}[1]{\ensuremath{\partial #1}}

\newcommand{\proofstd}{\ensuremath{\pi}}

\newcommand{\derivrulebinary}[3]{\frac{#1 \quad #2}{#3}}
\newcommand{\derivruleunary}[2]{\frac{#1}{#2}}

\DeclareFontFamily{OT1}{pzc}{}
\DeclareFontShape{OT1}{pzc}{m}{it}{<-> s * [1.200] pzcmi7t}{}
\DeclareMathAlphabet{\mathpzc}{OT1}{pzc}{m}{it}

\newcommand{\derivof}[4][\derives]
        {{\ensuremath{{#2} : {#3} \, {#1}\, {#4}}}}

\newcommand{\refof}[2]{\derivof{#1}{#2}{\falsenum}}
\renewcommand{\refof}[2]{\derivof{#1}{#2}{\emptycl}}
\renewcommand{\refof}[2]{\derivof{#1}{#2}{\bot}}

\newcommand{\deriveswithall}%
        {\vdash_{\!\!\!{\scriptscriptstyle \forall}}} 
\newcommand{\notderiveswithall}%
        {\nvdash_{\!\!\!{\scriptscriptstyle \forall}}} 

\newcommand{\resrule}[3]{\derivrulebinary{#1}{#2}{#3}}

\newcommand{\weakrule}[2]{\derivruleunary{#1}{#2}}

\newcommand{\clcfgtransitioncrammed}[2]%
        {\ensuremath{#1 \!\rightsquigarrow\! #2}}

\newcommand{\tvastd}{{\ensuremath{\alpha}}}

\newcommand{\fstd}{{\ensuremath{F}}}
\newcommand{\fvar}{{\ensuremath{G}}}
\newcommand{\falt}{\fvar}

\newcommand{\formf}{\ensuremath{F}}

\newcommand{\emptycl}{\bot}

\newcommand{\varx}{\ensuremath{x}}
\newcommand{\vary}{\ensuremath{y}}

\newcommand{\lita}{\ensuremath{a}}

\newcommand{\cla}{\ensuremath{A}}
\newcommand{\clb}{\ensuremath{B}}
\newcommand{\clc}{\ensuremath{C}}
\newcommand{\cld}{\ensuremath{D}}

\newcommand{\clausesetformat}[1]{\ensuremath{\mathbb{#1}}}
\newcommand{\clsc}{\clausesetformat{C}}
\newcommand{\clsd}{\clausesetformat{D}}

\newcommand{\setsofvarsorlitlarge}[2]%
        {\mathit{#1}\left({#2}\right)}
\newcommand{\setsofvarsorlit}[2]%
        {\mathit{#1}({#2})}
\newcommand{\setsofvarsorlitcompact}[2]%
        {\mathit{#1}\bigl({#2}\bigr)}

\newcommand{\setsofvarsorlitsup}[3]%
        {\mathit{#1}^{#2}({#3})}
\newcommand{\setsofvarsorlitsuplarge}[3]%
        {\mathit{#1}^{#2}\left({#3}\right)}
\newcommand{\setsofvarsorlitsupcompact}[3]%
        {\mathit{#1}^{#2}\bigl({#3}\bigr)}

\newcommand{\derivabbrev}[2]{\bigl( #1 \vdash #2 \bigr)}
\newcommand{\derivabbrevsmall}[2]{( #1 \vdash #2 )}
\newcommand{\derivabbrevcompact}[2]{\bigl( #1 \vdash #2 \bigr)}

\newcommand{\refutabbrevsmall}[1]{\derivabbrevsmall{#1}{\falsenum}}
\newcommand{\refutabbrevcompact}[1]{\derivabbrevcompact{#1}{\falsenum}}

\renewcommand{\refutabbrevsmall}[1]{\derivabbrevsmall{#1}{\!\emptycl}}
\renewcommand{\refutabbrevcompact}[1]{\derivabbrevcompact{#1}{\!\emptycl}}

\renewcommand{\refutabbrevsmall}[1]{\derivabbrevsmall{#1}{\!\bot}}
\renewcommand{\refutabbrevcompact}[1]{\derivabbrevcompact{#1}{\!\bot}}

\newcommand{\genericmeasure}[2]{{\mathit{#1}}_{#2}}

\newcommand{\genericformsmall}[2]{\mathit{#1}( #2 )}

\newcommand{\genericrefsmall}[3]%
    {{\mathit{#1}}_{#2}\refutabbrevsmall{#3}}
\newcommand{\genericrefcompact}[3]%
    {{\mathit{#1}}_{#2}\refutabbrevcompact{#3}}
\newcommand{\genericderiv}[4]%
    {{\mathit{#1}}_{#2}\derivabbrev{#3}{#4}}
\newcommand{\genericderivsmall}[4]%
    {{\mathit{#1}}_{#2}\derivabbrevsmall{#3}{#4}}
\newcommand{\genericderivcompact}[4]%
    {{\mathit{#1}}_{#2}\derivabbrevcompact{#3}{#4}}
\newcommand{\generictaut}[3]%
    {{\mathit{#1}}_{#2}\derivabbrev{}{#3}}
\newcommand{\generictautcompact}[3]%
    {{\mathit{#1}}_{#2}\derivabbrevcompact{}{#3}}
\newcommand{\generictautsmall}[3]%
    {{\mathit{#1}}_{#2}\derivabbrevsmall{}{#3}}
\newcommand{\length}[1][]{\genericmeasure{L}{#1}}

\newcommand{\lengthstd}{\length}
\newcommand{\lengthofarg}[1]{\genericformsmall{L}{#1}}

\newcommand{\lengthref}[2][]{\genericrefsmall{L}{#1}{#2}}

\newcommand{\widthstd}{\ensuremath{w}} 

\newcommand{\widthofarg}[2][]{\genericformsmall{W_{#1}}{#2}}

\newcommand{\widthref}[2][]{\genericrefsmall{W}{#1}{#2}}

\newcommand{\clspacestd}{s}
\newcommand{\clspaceof}[2][]{\genericformsmall{Sp_{#1}}{#2}}

\newcommand{\clspaceref}[2][]{\genericrefsmall{Sp}{#1}{#2}}

\newcommand{\varspaceof}[2][]{\genericformsmall{VarSp_{#1}}{#2}}

\newcommand{\varspaceref}[2][]{\genericrefsmall{VarSp}{#1}{#2}}

\newcommand{\formulaformat}[1]{\ensuremath{\mathit{#1}}}
\renewcommand{\formulaformat}[1]{\mathit{#1}}

\newcommand{\transitionarrow}{\rightsquigarrow}
\newcommand{\pebcfgtransition}[2]%
    {\ensuremath{#1 \transitionarrow #2}}
\newcommand{\pebcfgtransitionsqueeze}[2]%
    {#1 \! \transitionarrow \! #2}

\newcommand{\formatpebblingprice}[1]{\textsl{\textsf{#1}}}

\newcommand{\Pebblingprice}[1]%
    {\formatpebblingprice{Peb}\bigl(#1\bigr)}
\newcommand{\pebblingpricecompact}[1]%
    {\formatpebblingprice{Peb}\bigl(#1\bigr)}

\newcommand{\Bwpebblingprice}[1]%
    {\formatpebblingprice{BW-Peb}\bigl(#1\bigr)}
\newcommand{\bwpebblingpricecompact}[1]%
    {\formatpebblingprice{BW-Peb}\bigl(#1\bigr)}

\newcommand{\pebpersistentsymbol}{\bullet}
\newcommand{\pebvisitingsymbol}{\emptyset}

\newcommand{\bwpebpricepersistent}[1]%
    {\formatpebblingprice{BW-Peb}^{\pebpersistentsymbol}(#1)}
\newcommand{\Bwpebpricepersistent}[1]%
    {\formatpebblingprice{BW-Peb}^{\pebpersistentsymbol}\bigl(#1\bigr)}
\newcommand{\bwpebpricevisiting}[1]%
    {\formatpebblingprice{BW-Peb}^{\pebvisitingsymbol}(#1)}
\newcommand{\Bwpebpricevisiting}[1]%
    {\formatpebblingprice{BW-Peb}^{\pebvisitingsymbol}\bigl(#1\bigr)}

\newcommand{\pebpricepersistent}[1]%
    {\formatpebblingprice{Peb}^{\pebpersistentsymbol}(#1)}
\newcommand{\Pebpricepersistent}[1]%
    {\formatpebblingprice{Peb}^{\pebpersistentsymbol}\bigl(#1\bigr)}
\newcommand{\pebpricevisiting}[1]%
    {\formatpebblingprice{Peb}^{\pebvisitingsymbol}(#1)}
\newcommand{\Pebpricevisiting}[1]%
    {\formatpebblingprice{Peb}^{\pebvisitingsymbol}\bigl(#1\bigr)}

\newcommand{\bwpebblingpriceempty}[1]%
    {\formatpebblingprice{BW-Peb}^{\pebvisitingsymbol}(#1)}
\newcommand{\bwpebblingpriceemptycompact}[1]%
    {\formatpebblingprice{BW-Peb}^{\pebvisitingsymbol}\bigl(#1\bigr)}

\newcommand{\stoptime}{\tau}

\newcommand{\pebcontr}[2][G]{\ensuremath{\formulaformat{Peb}^{#2}_{#1}}}

\newcommand{\pebdeg}{\ensuremath{d}}

\newcommand{\pebaxcompact}[2]%
        [\pebdeg]{\ensuremath{\formulaformat{Ax}^{#1} \bigl(#2 \bigr)}}

\newcommand{\pqrxvar}[6]%
    {\ensuremath{\stdnot{\varx({#1})}_{#2} \lor \stdnot{\varx({#3})}_{#4} \lor %
    \sourceclausexvar[#6]{#5}}}

\newcommand{\pqr}[6]%
    {\ensuremath{\stdnot{#1}_{#2} \lor \stdnot{#3}_{#4} \lor %
    \sourceclausenodisplay[#6]{#5}}}
\newcommand{\pqrstd}{\pqr{p}{i}{q}{j}{r}{l}}
\newcommand{\pqrall}[6]%
        {\setdescrcompact
        {\pqr{#1}{#2}{#3}{#4}{#5}{#6}}{#2,#4 \in \intnfirst{\pebdeg}}}
\newcommand{\pqrallstd}%
        {\setdescrcompact{\pqrstd}{i,j \in \intnfirst{\pebdeg}}}

\newcommand{\sourceclausexvar}[2][n]%
        {\Lor_{#1 = 1}^{\pebdeg} \varx({#2})_{#1}}
\newcommand{\subsourceclausexvar}[3][n]%
        {\Lor_{#1 = {#2}}^{\pebdeg} \varx({#3})_{#1}}

\newcommand{\sourceclausexvarnodisplay}[2][n]%
        {\textstyle \Lor_{#1 = 1}^{\pebdeg} \varx({#2})_{#1}}

\newcommand{\sourceclausenodisplay}[2][n]%
        {\textstyle \Lor_{#1 = 1}^{\pebdeg} #2_{#1}}

\newcommand{\relativisation}[1]%
    {\ensuremath{\formulaformat{Rel}\bigl(#1 \bigr)}}

\newcommand{\extPHPnot}[2]
    {\ensuremath{\extendedversion{\formulaformat{PHP}}^{#1}_{#2}}}

\newcommand{\GraphOntoPHPnot}[1][G]%
    {\text{$\formulaformat{Onto}$-$\formulaformat{PHP}$}_{#1}}

\renewcommand{\extPHPnot}[2]%
    {\ephpnot{#1}{#2}}

\newcommand{\ephpnot}[2]%
    {\vphantom{\extendedversion{\formulaformat{PHP}}}
      {\smash{\extendedversion{\formulaformat{PHP}}}
        \vphantom{\formulaformat{PHP}}}^{#1}_{#2}}

\newcommand{\efphpnot}[2]%
    {\vphantom{\extendedversion{\formulaformat{FPHP}}}
      {\smash{\extendedversion{\formulaformat{FPHP}}}
        \vphantom{\formulaformat{FPHP}}}^{#1}_{#2}}

\newcommand{\ontophpnot}[2]%
    {\formulaformat{Onto}\text{-}\formulaformat{PHP}^{#1}_{#2}}

\newcommand{\ontofphpnot}[2]%
    {\formulaformat{Onto}\text{-}\formulaformat{FPHP}^{#1}_{#2}}

\newcommand{\extendedversion}[1]{\widetilde{#1}}

\makeatletter{}%
\usepackage{stmaryrd} 

\newcommand{\formatfunctiontosubconfiguration}[1]{\mathsf{#1}}
\newcommand{\formatfunctiontomulti}[1]{\mathcal{#1}}

\DeclareMathOperator{\dummystar}{*}
\newcommand{\pebblingcontrNT}[2][G]%
 {\ensuremath{\dummystar\!\!\formulaformat{Peb}^{#2}_{#1}}}

\newcommand{\somenodetrueclausedeg}[2]{\formulaformat{All}_{#1}^{+}({#2})}

\newcommand{\slashedstrickenletter}[1]{{\backslash\mkern-9mu #1}}
            
\newcommand{\strikethroughcommand}[1]{\slashedstrickenletter{#1}}

\newcommand{\abovevertices}[2][G]%
    {{#1}_{#2}^{\hspace{-0.2 pt}\triangledown}}
\newcommand{\aboveverticesNR}[2][G]%
    {{#1}_{\strikethroughcommand{#2}}^{\hspace{-0.3 pt}\triangledown}}
\newcommand{\belowvertices}[2][G]%
    {{#1}^{#2}_{\hspace{-0.6 pt}\vartriangle}}
\newcommand{\belowverticesNR}[2][G]%
    {{#1}^{\strikethroughcommand{#2}}_{\hspace{-0.6 pt}\vartriangle}}

\newcommand{\lpebblingpricecompact}[1]%
    {\formatpebblingprice{L-Peb}\bigl(#1\bigr)}

\newcommand{\scnot}[2]{#1 \langle #2 \rangle}
\newcommand{\scnotcompact}[2]{#1 \bigl\langle #2 \bigr\rangle}

\newcommand{\spcanonconfcompact}[1]%
        {\formatfunctiontosubconfiguration{canon}\bigl({#1}\bigr)}

\newcommand{\spprojsubsub}[4]%
    {\formatfunctiontosubconfiguration{proj}_{\scnot{#1}{#2}}(\scnot{#3}{#4})}
\newcommand{\spprojsubsubcompact}[4]%
    {\formatfunctiontosubconfiguration{proj}_{\scnot{#1}{#2}}%
    \bigl(\scnot{#3}{#4}\bigr)}
\newcommand{\spprojsubconf}[3]%
    {\formatfunctiontosubconfiguration{proj}_{\scnot{#1}{#2}}({#3})}
\newcommand{\spprojsubconfcompact}[3]%
    {\formatfunctiontosubconfiguration{proj}_{\scnot{#1}{#2}}\bigl({#3}\bigr)}
\newcommand{\spprojconfsub}[3]%
    {\formatfunctiontosubconfiguration{proj}_{#1}(\scnot{#2}{#3})}
\newcommand{\spprojconfsubcompact}[3]%
    {\formatfunctiontosubconfiguration{proj}_{#1}\bigl(\scnot{#2}{#3}\bigr)}
\newcommand{\spprojconfconf}[2]%
    {\formatfunctiontosubconfiguration{proj}_{#1}({#2})}
\newcommand{\spprojconfconfcompact}[2]%
    {\formatfunctiontosubconfiguration{proj}_{#1}\bigl({#2}\bigr)}

\newcommand{\spclossubcompact}[2]%
        {\formatfunctiontoset{cl}\bigl(\scnotcompact{#1}{#2}\bigr)}

\newcommand{\spintersubcompact}[2]%
        {\formatfunctiontoset{int}\bigl(\scnotcompact{#1}{#2}\bigr)}

\newcommand{\spcoversubcompact}[2]%
        {\formatfunctiontoset{cover}\bigl(\scnotcompact{#1}{#2}\bigr)}

\newcommand{\spcoverconfcompact}[1]%
        {\formatfunctiontoset{cover}\bigl({#1}\bigr)}

\newcommand{\spinducedblack}[1]%
    {\formatfunctiontoset{Bl} (#1)}
\newcommand{\spinducedwhite}[1]%
    {\formatfunctiontoset{Wh} (#1)}
\newcommand{\spinducedblackcompact}[1]%
    {\formatfunctiontoset{Bl} \bigl(#1 \bigr)}
\newcommand{\spinducedwhitecompact}[1]%
    {\formatfunctiontoset{Wh} \bigl(#1 \bigr)}

\newcommand{\pathclausedeg}[2][\pebdeg]%
    {\somenodetrueclausedeg[#1]{\vertexpath{#2}}}
\newcommand{\pathclauseNRdeg}[2][\pebdeg]%
    {\somenodetrueclausedeg[#1]{\vertexpathNR{#2}}}

\newcommand{\blacktruthdegexplicit}[4]%
        {\setdescrcompact
        {{\textstyle \Lor_{#2 = 1}^{#3} {#1}_{#2}}}
        {{#1} \in {#4}}}

\newcommand{\binsubtree}[1]{T^{#1}}

\newcommand{\vertexpath}[1]{{P}^{#1}}
\newcommand{\vertexpathNR}[1]{{P}_{*}^{#1}}

\newcommand{\unrelatedNP}[1]%
        {T \setminus \bigl(\binsubtree{#1} \unionSP \vertexpath{#1} \bigr)}
\newcommand{\unrelatedsmallNP}[1]%
        {T \setminus (\binsubtree{#1} \unionSP \vertexpath{#1} )}

\newcommand{\abovelevelblockerminsizecompact}%
    [2]{L_{\succeq{#1}}\bigl({#2}\bigr)}

\newcommand{\necessaryhidingvert}[2]%
{{#1}{\scriptstyle{\llfloor {#2} \rrfloor}}}

\newcommand{\Klawepropertyprefix}{Limited hiding-cardinality\xspace}

\newcommand{\klawepropacronym}{LHC property\xspace}

\newcommand{\nongenklaweprop}%
{non-generalized \Klawepropertyprefix property\xspace}
\newcommand{\nongenklawepropacronym}%
{non-generalized \klawepropacronym}
\newcommand{\nongenklawepropacronymWithParam}%
{(non-generalized) \klawepropacronym}

\newcommand{\siblingnonreachabiblitypropertynoref}%
{Sibling non-reachability property\xspace}
\newcommand{\Siblingnonreachabiblitypropertynoref}%
{Sibling non-reachability property\xspace}
\newcommand{\siblingnonreachabiblityproperty}%
{\siblingnonreachabiblitypropertynoref~%
\ref{property:sibling-non-reachability-property}\xspace}
\newcommand{\Siblingnonreachabiblityproperty}%
{\Siblingnonreachabiblitypropertynoref~%
\ref{property:sibling-non-reachability-property}\xspace}

\newcommand{\introducetermanmpctext}%
    {a \introduceterm{\mpctext{}}\xspace}

\newcommand{\introducetermamultipebblingtext}%
  {a \introduceterm{\multipebblingtext{}}\xspace}

\newcommand{\blobpebblingtext}{blob-pebbling\xspace}

\newcommand{\multipebblingtext}{\blobpebblingtext}

\newcommand{\mpcostblack}[1]%
        {\formatpebblingprice{cost}_{\mpcblacks}( #1 )}
\newcommand{\mpcostwhite}[1]%
        {\formatpebblingprice{cost}_{\mpcwhites}( #1 )}

\newcommand{\blobpebblingpricecompact}[1]%
    {\formatpebblingprice{Blob-Peb}\bigl(#1\bigr)}

\newcommand{\multipebblingpricecompact}[1]%
    {\formatpebblingprice{Blob-Peb}\bigl(#1\bigr)}

\newcommand{\mpcblacks}{\formatfunctiontomulti{B}}
\newcommand{\mpcwhites}{\formatfunctiontomulti{W}}

\newcommand{\mpscnotcompact}[2]%
        {\big[ {#1} \big] \bigl\langle {#2} \bigr\rangle}

\newcommand{\mpctext}{\blobpebblingtext con\-fig\-u\-ra\-tion\xspace}

\newcommand{\chargeablevertices}[1]%
{\formatfunctiontoset{chargeable}({#1}) }
\newcommand{\chargeableverticescompact}[1]%
{\formatfunctiontoset{chargeable}\bigl({#1}\bigr) }

\newcommand{\blackschargedfor}[1][]%
    {\mpcblacks_{#1}}
\newcommand{\whiteschargedfor}[1][]%
    {\mpcwhites_{#1}^{\hspace{-0.3 pt}\vartriangle}}

\newcommand{\whitesbelowjustblocked}%
    {\mpcwhites_{B}^{\hspace{-0.3 pt}\vartriangle}}
\newcommand{\whitesbelowhidden}%
    {\mpcwhites_{H}^{\hspace{-0.3 pt}\vartriangle}}

\newcommand{\whitestight}%
    {\mpcwhites_{T}^{\hspace{-0.3 pt}\vartriangle}}

\makeatletter{}%
\newcommand{\auxkerset}{K}
\newcommand{\auxkersetstar}[1]{K^{#1}}

\newcommand{\cltrans}{\cld}
\newcommand{\cltranssubst}{\cld\substituted}
\newcommand{\clother}{\clb}

\newcommand{\claxiom}{\cla}
\newcommand{\claxiomsubst}{\cla\substituted}

\newcommand{\graphg}{\mathcal{D}}

\newcommand{\simfalsesubst}{simultaneous falsifiability\xspace}
\newcommand{\simfalse}{simultaneously falsifiable\xspace}
\newcommand{\simfalseargs}[2]%
    {\ensuremath{#1} and~\ensuremath{#2} are \simfalse}
\newcommand{\simfalseargscertainly}[2]%
    {\ensuremath{#1} and~\ensuremath{#2} are certainly \simfalse}
\newcommand{\simfalseargsmathtext}[2]%
    {\text{ \simfalseargs{#1}{#2}}}

\newcommand{\nbhd}{\mathcal N}%

\newcommand{\substform}[2]{{#1}[{#2}]}

\newcommand{\clspacesym}{\ensuremath{\clspacestd}}
\newcommand{\homogeneous}{homogeneous\xspace}
\newcommand{\homogeneity}{homogeneity\xspace}

\newcommand{\confindex}{t} %
\newcommand{\maxclscindex}{\tau} %

\newcommand{\orignvars}{N}
\newcommand{\newnvars}{n}

\newcommand{\indexnalt}{\newnvars}
\newcommand{\leftsize}{\orignvars}
\newcommand{\rightsize}{\newnvars}
\newcommand{\orignvarszero}{\orignvars_0}
\newcommand{\newnvarszero}{\newnvars_0}
\newcommand{\rightsizezero}{\newnvarszero}

\newcommand{\boundaryexp}[3]%
    {$({#2},{#1},{#3})$\nobreakdash-boundary expander\xspace}
\newcommand{\nmboundaryexp}[5]%
    {${#1}\times{#2}$ \boundaryexp{#3}{#4}{#5}}
\newcommand{\boundaryexpnodeg}[2]%
    {$({#1},{#2})$\nobreakdash-boundary expander\xspace}
\newcommand{\nmboundaryexpnodeg}[4]{${#1}\times{#2}$ \boundaryexpnodeg{#3}{#4}}

\newcommand{\expguarantee}{\expansionguarantee}
\newcommand{\expfactor}{\expansionfactor}
\newcommand{\expdegree}{\expanderdegree}
\newcommand{\lsize}{\leftsize}
\newcommand{\rsize}{\rightsize}

\newcommand{\boundaryexpnodegstd}%
  {\boundaryexpnodeg{\expguarantee}{\expfactor}}
\newcommand{\boundaryexpstd}%
  {\boundaryexp{\expdegree}{\expguarantee}{\expfactor}}
\newcommand{\nmboundaryexpstd}%
  {\nmboundaryexp{\lsize}{\rsize}{\expdegree}{\expguarantee}{\expfactor}}
\newcommand{\nmboundaryexpnodegstd}%
  {\nmboundaryexpnodeg{\lsize}{\rsize}{\expguarantee}{\expfactor}}

\newcommand{\aboundaryexpnodegstd}%
  {an \boundaryexpnodeg{\expguarantee}{\expfactor}}
\newcommand{\aboundaryexpstd}%
  {an \boundaryexp{\expdegree}{\expguarantee}{\expfactor}}
\newcommand{\annmboundaryexpstd}%
  {an \nmboundaryexp{\lsize}{\rsize}{\expdegree}{\expguarantee}{\expfactor}}
\newcommand{\annmboundaryexpnodegstd}%
  {an \nmboundaryexpnodeg{\lsize}{\rsize}{\expguarantee}{\expfactor}}

\newcommand{\Aboundaryexpnodegstd}%
  {An \boundaryexpnodeg{\expguarantee}{\expfactor}}
\newcommand{\Aboundaryexpstd}%
  {An \boundaryexp{\expdegree}{\expguarantee}{\expfactor}}
\newcommand{\Annmboundaryexpstd}%
  {An \nmboundaryexp{\lsize}{\rsize}{\expdegree}{\expguarantee}{\expfactor}}
\newcommand{\Annmboundaryexpnodegstd}%
  {An \nmboundaryexpnodeg{\lsize}{\rsize}{\expguarantee}{\expfactor}}

\newcommand{\defi}{:=}

\newcommand{\XOR}{XOR\xspace}
\newcommand{\XORification}{XOR\-ification\xspace}
\newcommand{\xorification}{\XORification}

\newcommand{\xorified}{XORified\xspace}
\newcommand{\xorsubstitution}{\XOR substitution\xspace}

\newcommand{\CNF}{CNF\xspace}

\newcommand{\graph}{\mathcal G}
\newcommand{\expandergraph}{\mathcal{G}}

\newcommand{\initialcnfwidth}{k}

\newcommand{\widthmax}{\widthstd}

\newcommand{\spacelowerbound}{\clspacesym}

\newcommand{\leftvertexset}{U}
\newcommand{\leftvertexsubset}{U'}
\newcommand{\rightvertexset}{V}
\newcommand{\rightvertexsubset}{V'}
\newcommand{\expansionguarantee}{r}
\newcommand{\expansionfactor}{c}
\newcommand{\expanderdegree}{d}
\newcommand{\Ker}{\operatorname{Ker}}
\renewcommand{\Ker}{\operatorname{ker}}
\newcommand{\closure}{\gamma}

\newcommand{\expandersubgraph}[2]{#1\setminus #2}
\newcommand{\substituted}{[\expandergraph]}

\newcommand{\clauseset}{\mathcal{G}^{-1}}
\newcommand{\clsetwithindex}[1][]{\clauseset}
\renewcommand{\clsetwithindex}[1][]{\clsd^{#1}}
\newcommand{\variables}{\mathit{Vars}}
\newcommand{\refofsubst}{\proofstd}
\newcommand{\refoforig}{\proofstd'}
\newcommand{\formsubst}{\formf\substituted}
\newcommand{\formorig}{\formf}

\newcommand{\refutationstd}{\proofstd}

\newcommand{\smallepsilon}{\varepsilon}
\newcommand{\smalldelta}{\delta}
\newcommand{\mindegree}{\expanderdegree_0}
\newcommand{\smallepsilonalt}{\varepsilon'}

\newcommand{\widthlower}{\ell}

\makeatletter{}%

\newcommand{\theauthorCB}{the first author\xspace}

\newcommand{\theauthorJN}{the second author\xspace}

\ifthenelse{\boolean{conferenceversion}}
{
  \newcommand{\mysubsection}[1]{\subparagraph*{#1}}

}
{
  \newcommand{\mysubsection}[1]{\subsection{#1}}

}

\ifthenelse{\boolean{conferenceversion}}
{}
{
\newtheoremstyle{metacommenttheoremstyle}%
    {3pt}%
    {3pt}%
    {\sffamily \itshape \scriptsize
    }%
    {}%
    {\bfseries \scshape \footnotesize }%
    {:}%
    { }%
    {}%

\theoremstyle{metacommenttheoremstyle}
}
\newtheorem{jncommentcontainer}{Jakob's comment}
\newtheorem{cbcommentcontainer}{Christoph's comment}

\newcounter{rbcounter}
\newcommand{\randbem}[3]%
{\stepcounter{rbcounter}%
  \parbox[t]{0mm}{$^{\arabic{rbcounter}}$}%
  \marginpar%
  {\textcolor{#1}%
    {\raggedright\footnotesize$\mathbf{#2}^{\arabic{rbcounter}}$: #3}%
  }
}

\ifthenelse{\isundefined{\DONOTINSERTCOMMENTS}}
{ %

  \newcommand{\jncomment}[1]%
  {\begin{jncommentcontainer} \textcolor{blue}{#1} \end{jncommentcontainer}}

  \newcommand{\cbcomment}[1]%
  {\begin{cbcommentcontainer} \textcolor{magenta}{#1} \end{cbcommentcontainer}}

  \newcommand{\jn}[1]{\randbem{blue}{J}{#1}}
  \newcommand{\chr}[1]{\randbem{magenta}{C}{#1}}

}
{ %
  \newcommand{\jncomment}[1]{}
  \newcommand{\cbcomment}[1]{}
  \newcommand{\chr}[1]{}
  \newcommand{\jn}[1]{}
}

\numberwithin{equation}{section}

\allowdisplaybreaks
\begin{document}

\title{Supercritical Space-Width Trade-offs for Resolution%
    \thanks{This is the full-length version of the paper with the same title 
      that appeared in \emph{Proceedings of the 43rd International Colloquium
        on Automata, Languages and Programming ({ICALP}~'16)}.}}

\author{%
  Christoph Berkholz \\
  Humboldt-Universität zu Berlin
  \and
  Jakob Nordström \\
  KTH Royal Institute of Technology}

\date{\today}

\maketitle

\thispagestyle{empty}

\pagestyle{fancy}
\fancyhead{}
\fancyfoot{}
\renewcommand{\headrulewidth}{0pt}
\renewcommand{\footrulewidth}{0pt}

\fancyhead[CE]{\slshape 
  SUPERCRITICAL SPACE-WIDTH TRADE-OFFS FOR RESOLUTION}
\fancyhead[CO]{\slshape \nouppercase{\leftmark}}
\fancyfoot[C]{\thepage}

\setlength{\headheight}{13.6pt}

\makeatletter{}%

\begin{abstract}  
  We show that there are CNF formulas which can be refuted in
  resolution in both small space and small width, but for which any
  small-width 
  proof must have space exceeding by far the
  linear worst-case upper bound. This significantly strengthens the
  space-width trade-offs \mbox{in [Ben-Sasson~'09]}, and provides one more
  example 
  of trade-offs in the ``supercritical''
  regime above worst case recently identified by [Razborov~'16]. We
  obtain our results by using Razborov's new hardness condensation
  technique and combining it with the space lower bounds in
  [Ben-Sasson and Nordström~'08].
\end{abstract}

\makeatletter{}%

\section{Introduction}
\label{sec:intro}

Propositional proof complexity studies the problem of how to provide
concise, polynomial-time checkable certificates that formulas in
conjunctive normal form (CNF) are
unsatisfiable. Research in this area was initiated
in~\cite{CR79Relative} as a way of attacking the problem of showing
that
$\NP \neq \coNP$,
and hence
$\Pclass \neq \NP$,
and it is therefore natural that the main focus has been on proving
upper and lower bounds on proof length/size. 
More recently, however, other complexity measures have also been
investigated, and this study has revealed a rich and often surprising
web of connections.

\mysubsection{Resolution Length, Width, and Space}

Arguably the most thoroughly studied proof system in proof complexity is
\introduceterm{resolution}, 
which appeared in~\cite{Blake37Thesis}
and began to be investigated in connection with automated theorem proving 
in the 1960s
\cite{DLL62MachineProgram,
  DP60ComputingProcedure,
  Robinson65Machine-oriented}.
\ifthenelse{\boolean{conferenceversion}}
{Because of its simplicity this proof system}
{Because of its simplicity---there is only one derivation rule---and
  because all lines in a proof are clauses, this proof system}
is well suited for proof search, and it lies at the heart of current
state-of-the-art SAT solvers based on so-called 
\introduceterm{conflict-driven clause learning} 
\cite{BS97UsingCSP,MS99Grasp,MMZZM01Engineering}.

It is not hard to show that any unsatisfiable CNF formula over
$\nvars$~variables can be proven unsatisfiable, or \introduceterm{refuted},
by a resolution refutation containing $\exp(\bigoh{\nvars})$~clauses,
and this holds even in the restricted setting of
\introduceterm{tree-like resolution}, where each intermediate clause
in the refutation has to be rederived from scratch every time it is used.
In the breakthrough paper~\cite{Haken85Intractability}, Haken obtained
a length lower bound on the form 
$\exp \bigl( \Bigomega{\nvars^\delta}\bigr)$ 
for general resolution refutations of so-called pigeonhole principle
formulas, and  
this result
was later followed by truly exponential lower bounds 
$\exp ( \bigomega{\nvars} )$ 
\ifthenelse{\boolean{conferenceversion}}
{for other formula families in, e.g.,  
  \cite{BKPS02Efficiency,CS88ManyHard,Urquhart87HardExamples}.}
{for other formula families in 
  \cite{Urquhart87HardExamples,CS88ManyHard,BKPS02Efficiency}
  and many other papers.}

In a seminal paper~\cite{BW01ShortProofs}, Ben-Sasson and Wigderson
identified \introduceterm{width}, measured as the largest size of any
clause appearing in a refutation, as another interesting complexity
measure for resolution. Clearly, any unsatisfiable CNF formula over
$\nvars$~variables can be refuted in width at most~$\nvars$. Moreover, any 
\ifthenelse{\boolean{conferenceversion}}
{refutation} 
{resolution refutation} 
in width~$\widthstd$ need never be longer
than~$\nvars^{\bigoh{\widthstd}}$, since this is an upper bound on the
number of distinct clauses of width~$\widthstd$ (and 
this naive counting argument is essentially tight~\cite{ALN16NarrowProofs}).
What Ben-Sasson and Wigderson showed is that strong enough
\emph{lower} bounds on width also imply lower bounds on length; in
particular that linear $\bigomega{\nvars}$ width lower bounds imply
exponential $\exp (\bigomega{\nvars})$ length lower bounds
for CNF formulas of bounded width.
This connection can be used to rederive almost all currently known
resolution length lower bounds.

Motivated by questions in SAT solving, where efficient memory
management is a major concern, a more recent line of research in proof
complexity has examined a third complexity measure on proofs, namely
\introduceterm{space}.
This study was initiated by Esteban and
Torán~\cite{ET01SpaceBounds}, who defined the 
\introduceterm{(clause) space} of a resolution proof as the maximal
number of clauses needed to be kept in memory during verification of
\ifthenelse{\boolean{conferenceversion}}
{the proof.}%
{the proof,}%
  \footnote{For completeness, we want to mention that for resolution there 
  is also a \introduceterm{total space} measure counting the total number 
  of literals in memory (with repetitions), which has been studied in
  \cite{ABRW02SpaceComplexity,BGT14TotalSpace,BBGHMW15SpaceProofCplx,Bonacina16TotalSpace}.
  In this paper, however, ``space'' will always mean ``clause space''
  in the sense of~\cite{ET01SpaceBounds} unless otherwise stated.}
\ifthenelse{\boolean{conferenceversion}}
{It can be shown that} 
{a definition that was generalized to other proof systems by
  Alekhnovich \etal~\cite{ABRW02SpaceComplexity}.
  It should be noted that although the original impetus for
  investigating proof space came from the applied SAT solving side,
  space complexity is of course a well-studied measure in its own right
  in computational complexity, and the study of space in proof
  complexity has turned out to be of intrinsic interest in that it has
  uncovered intriguing connections to proof length and width. 
  It can be shown that} 
a CNF formula over $\nvars$~variables can always be refuted
in space~$\nvars + \bigoh{1}$ even in tree-like
resolution~\cite{ET01SpaceBounds},
although the refutation thus obtained might have exponential
length. Linear space lower bounds   
matching the worst-case upper bound up to constant
factors were obtained for a number of formula families in
\ifthenelse{\boolean{conferenceversion}}
{\cite{ABRW02SpaceComplexity,BG03SpaceComplexity,ET01SpaceBounds}.}
{\cite{ET01SpaceBounds,ABRW02SpaceComplexity,BG03SpaceComplexity}.}

The space lower bounds obtained in the papers just discussed turned
out to match closely known lower bounds on width, and in a
strikingly simple and beautiful result Atserias and
Dalmau~\cite{AD08CombinatoricalCharacterization} showed that in fact 
\ifthenelse{\boolean{conferenceversion}}
{the resolution width of refuting a \kcnfform~$F$
  provides a lower bound for the clause space  required.%
  \footnote{Note that this is a nontrivial connection since a lower bound
    on width, \ie the \emph{number of literals} in a clause, is shown to
    imply essentially the same lower bound on the
    \emph{number of clauses} needed.}}
{the resolution width of refuting a \kcnfform~$F$
  is a lower bound on the clause space  required,%
  \footnote{Note that this is a nontrivial connection since a lower bound
    on width, \ie the \emph{number of literals} in a clause, is shown to
    imply essentially the same lower bound on the
    \emph{number of clauses} needed.}
  minus an additive term~$\clwidth$ adjusting for the largest width of
  any clause in~$F$.}
This allows to recover the space lower bounds mentioned above as
immediate consequences of width lower bounds shown
in~\cite{BW01ShortProofs}.
Furthermore, it follows from~\cite{AD08CombinatoricalCharacterization}
that for $\clwidth = \bigoh{1}$ 
any \kcnfform that can be refuted by just
keeping a constant number of clauses in memory can also be refuted in
polynomial length and constant width. 
These connections go only in one direction, however---in the sequence of papers
\cite{Nordstrom09NarrowProofsSICOMP,NH13TowardsOptimalSeparation,BN08ShortProofs}
it was shown that there are formula families that have high space
complexity although they possess refutations in linear length and
constant width.

\mysubsection{Resolution Trade-offs}

As was discussed above, a resolution proof in sufficiently small width
will by necessity also be short, whereas the linear worst-case
upper bound on space is achieved by a proof in exponential length. It
is natural to ask, therefore,
whether for a given  formula~$F$
there exists a single 
\ifthenelse{\boolean{conferenceversion}}
{refutation that can}
{resolution refutation of~$F$ that can}
simultaneously optimize these different complexity measures.
The question of trade-offs between proof complexity measures was first raised by
Ben-Sasson~\cite{Ben-Sasson09SizeSpaceTradeoffs}, who gave a strong
negative answer for space versus width. 
More precisely, what was shown in~\cite{Ben-Sasson09SizeSpaceTradeoffs} is 
that there are formulas which are refutable separately in constant
width and in constant space, but for which any resolution proof
minimizing one of the measures must exhibit almost worst-case linear
behaviour with respect to 
the other measure.

A question that arises in the context of SAT solving is whether it is
possible to simultaneously optimize size and space (corresponding to
running time and memory usage).
In addition to the space-width trade-offs discussed above, 
in~\cite{Ben-Sasson09SizeSpaceTradeoffs} Ben-Sasson also proved a
size-space trade-off for the subsystem tree-like resolution, and building on
\ifthenelse{\boolean{conferenceversion}}
{\cite{Ben-Sasson09SizeSpaceTradeoffs,BN08ShortProofs}} 
{\cite{BN08ShortProofs,Ben-Sasson09SizeSpaceTradeoffs}} 
it was shown in~\cite{BN11UnderstandingSpace}
for general resolution 
that there are formulas which have refutations in linear length and
also in small space, but for which any space-efficient refutation must
have superpolynomial or even exponential length.
\ifthenelse{\boolean{conferenceversion}}
{Beame \etal~\cite{BBI16TimeSpace} and Beck \etal~\cite{BNT13SomeTradeoffs}
  exhibited formulas
  over $\nvars$~variables refutable in length polynomial in~$\nvars$
  where bringing the space down to linear, 
  or even just shaving a constant factor of the polynomial space bound
  that follows immediately from the length bound,
  incurs a superpolynomial penalty in proof length.}
{Beame \etal~\cite{BBI16TimeSpace} 
  extended the range of parameters of the trade-offs further by
  exhibiting formulas  
  over $\nvars$~variables refutable in length polynomial in~$\nvars$
  where bringing the space down to linear, 
  or even just shaving a constant factor of the polynomial space bound
  that follows immediately from the length bound,
  incurs a superpolynomial penalty in proof length, 
  a result that was generalized and strengthened in~\cite{BNT13SomeTradeoffs}.} 

Turning finally to the relation between length and width, 
what was shown in \cite{BW01ShortProofs}
is that a 
\ifthenelse{\boolean{conferenceversion}}
{short refutation} 
{short resolution refutation} 
can be converted to a 
refutation of small width,
but the way this conversion is done in~\cite{BW01ShortProofs} blows up
the length exponentially.  
Thapen~\cite{Thapen14Trade-off} proved that this 
is inherent by exhibiting formulas refutable in small width
and small length, but for which any small-width refutation has to have
exponential length.
For the restricted case of tree-like resolution,
Razborov~\cite{Razborov16NewKind} 
recently showed that there are
formulas refutable in small width for which any tree-like refutation
even doing slightly better than the trivial linear upper bound with
respect to width must by necessity have doubly exponential length.

\ifthenelse{\boolean{conferenceversion}}
{We want to highlight} 
{We want to emphasize} 
an intriguing property of the trade-off results in
\cite{BBI16TimeSpace,BNT13SomeTradeoffs,Razborov16NewKind}
\ifthenelse{\boolean{conferenceversion}}
{that sets them apart from}
{that was highlighted by Razborov, and that sets these results apart from}
the other trade-offs surveyed above.  Namely, for most trade-off results
between complexity measures it is the case that the trade-off plays out in
the region between the worst-case upper bounds for the measures,
where  as one measure decreases the other measure has to
approach its critical worst-case value.
However, the short resolution proofs in 
\cite{BBI16TimeSpace,BNT13SomeTradeoffs}
require space even polynomially larger than the 
worst-case 
upper bound, and the
small-width
tree-like proofs in~\cite{Razborov16NewKind}
require proofs exponentially larger than the
exponential upper bound for tree-like length.
To underscore the dramatic nature of such trade-off results,
Razborov refers to them as \introduceterm{ultimate}
in the preliminary version~\cite{Razborov15Ultimate}
of~\cite{Razborov16NewKind}. 
In this paper, we will instead use the term 
\introduceterm{supercritical trade-offs}, 
which we feel better describes the behaviour that one of the
complexity measures is pushed up into the supercritical regime above
worst case when the other measure is decreased.

\mysubsection{Our Contribution}

Answering Razborov's call in \cite{Razborov16NewKind} for more
examples of the 
type of trade-offs discussed above,
in this paper
we prove a supercritical trade-off between space and width in
resolution. 
As already observed, any refutation in width~$\widthstd$ of a CNF
formula over $\nvars$~variables in general resolution need not contain
more than $\bigoh{\nvars^{\widthstd}}$ clauses, which is also a
trivial upper bound on the space complexity of such a refutation. Our
main result is that this bound is essentially tight, and is also somewhat
robust. Namely, we show that there are $\nvars$\nobreakdash-variable
formulas that can be refuted in width~$\widthstd$,
but for which any refutation in width 
even up to almost a multiplicative logarithmic factor larger than this
requires space~$\nvars^{\bigomega{\widthstd}}$.

\begin{theorem}
  \label{thm:maintheorem}
  For any constant $\varepsilon>0$ and any non-decreasing function
  $\widthlower(\indexnalt)$,
  $6  \leq 
  \widthlower(\indexnalt)  \leq
  \indexnalt^{\frac{1}{2}-\epsilon}$,
  there is a
  family $\set{\formf_\indexnalt}_{\indexnalt\in \mathbb N}$ of
  $\indexnalt$-variable \CNF{} formulas which can be
  \ifthenelse{\boolean{conferenceversion}} 
  {refuted in resolution width~$\widthlower(\indexnalt)$
    but for which any  refutation in width}
  {refuted in resolution in width~$\widthlower(\indexnalt)$
    but for which any  resolution refutation in width}
  $\littleoh{\widthlower(\indexnalt)\log \indexnalt}$ 
  \ifthenelse{\boolean{conferenceversion}}
  {requires space}
  {requires clause space at least}
  $\indexnalt^{\bigomega{\widthlower(\indexnalt)}}$.
\end{theorem}

\mysubsection{Techniques}

In one sentence, we obtain our results by using Razborov's hardness
condensation technique in \cite{Razborov16NewKind} and combining it
with the space lower bounds in~\cite{BN08ShortProofs}.

In slightly more detail, our starting point are the so-called
\introduceterm{pebbling formulas} defined in~\cite{BW01ShortProofs}.
These formulas are refutable in constant width, but
it was observed in~\cite{Ben-Sasson09SizeSpaceTradeoffs} that space lower bounds
for pebble games on directed acyclic graphs (DAGs) carry over to lower
bounds on the \emph{number of variables} kept simultaneously in memory
in resolution refutations of pebbling formulas defined over these DAGs.
It was shown in~\cite{BN08ShortProofs} that substituting every
variable in such formulas by an exclusive or of two new variables and
expanding out to CNF produces a new family of formulas which are
still refutable in constant width but for which the variable space
lower bounds have been amplified to clause space lower bounds.

The 
result
in~\cite{BN08ShortProofs} is one of several examples of how 
\introduceterm{\xorsubstitution{}},
or
\introduceterm{\xorification{}},
has been used to amplify weak proof complexity lower bounds to much
stronger lower bounds. In all of these applications distinct variables
of the original formula are replaced by 
disjoint sets of new variables.  The wonderfully simple (with hindsight) but
powerful new idea in~\cite{Razborov16NewKind} is 
to instead do \xorsubstitution with overlapping sets of variables from
a much smaller variable pool (but with exclusive ors of higher arity).

This recycling of variables has the consequence that hardness
amplification as in~\cite{BN08ShortProofs} no longer works, since it
crucially depends on the fact that all new substitution variables are
distinct. What Razborov showed in~\cite{Razborov16NewKind} was
essentially that if the pattern of overlapping variable substitutions is
described by a strong enough bipartite expander, then locally there
are enough distinct new variables to make tree-like amplification
lower bounds as in~\cite{Ben-Sasson09SizeSpaceTradeoffs} go through
over a fairly wide range of the parameter space, yielding 
supercritical 
trade-offs between width and tree-like length. 
Since in addition the number of variables in the formula has decreased
significantly, this can be viewed as a kind of 
\introduceterm{hardness condensation}. 

We use Razborov's idea of \xorification with recycled variables,
but since we want to obtain results not for tree-like but
for DAG-like resolution the technical details of our proofs are  somewhat
different. At a high level, we start with formulas over
$\orignvars$~variables that are refutable in constant width but
require space 
$\bigomega{\orignvars / \log \orignvars}$.
We modify these formulas by applying \mbox{$\widthstd$-wise}
\xorification using a much smaller set of $\nvars$~variables,
and then show
that from any refutation in width~$\bigoh{\widthstd}$ of this new,
\xorified formula it is possible to recover a refutation of the
original formula with comparable space complexity. But this means that
any small-width refutation of the \xorified formula must have space
complexity \mbox{roughly $\bigomega{\orignvars / \log \orignvars}$}. Choosing
parameters so that $\orignvars \approx \nvars^\widthstd$ yields the
bound stated in \refthm{thm:maintheorem}.

We should point out that compared to~\cite{Razborov16NewKind} we get
significantly less robust trade-offs, which break down already for a
multiplicative logarithmic increase in width. 
This is mainly due to the fact that we deal not with tree-like
resolution as in~\cite{Razborov16NewKind}, 
but with the much stronger general resolution proof system producing
DAG-like proofs.
We share with~\cite{Razborov16NewKind} the less desirable feature
that although our formulas only have $\nvars$~variables they contain
on the order of~$\nvars^\widthstd$ clauses. Thus, measured in terms of
formula size our space-width trade-offs do not improve
on~\cite{Ben-Sasson09SizeSpaceTradeoffs}, and the width of our
formulas is not constant but scales linearly with~$\widthstd$. 
Still, since the number of variables provides a worst-case upper bound
on space (independently of formula size), measured in terms of
variables it seems fair to say that the trade-off result in
\refthm{thm:maintheorem} is fairly dramatic.

\mysubsection{Organization of This Paper}

The rest of this paper is organized as follows.
We start by reviewing some preliminaries in
\refsec{sec:prelims}.
In \refsec{sec:proof-main-thm} we prove our main result assuming a
hardness condensation lemma, and this lemma is then established in
\refsec{sec:hardness-condensation}. 
\ifthenelse{\boolean{conferenceversion}}
{We conclude in
  \refsec{sec:conclusion} with a discussion of possible directions for
  future research.
  Due to space constraints, we omit some of the
  proofs in this 
  extended abstract,
  referring the reader to the
  upcoming full-length version for the missing details.}
{We conclude in
  \refsec{sec:conclusion} with a discussion of possible directions for
  future research.
  For completeness, proofs of some technical claims are provided in
  \refapp{app:appendix}.}

\makeatletter{}%

\section{Preliminaries}
\label{sec:prelims}

A \introduceterm{literal} over a Boolean variable $\varx$ is either
the variable $\varx$ itself (a \introduceterm{positive literal}) or
its negation~$\olnot{\varx}$ (a~\introduceterm{negative literal}).
We define
$\olnot{\olnot{\varx}} = \varx$.
A \introduceterm{clause}
$\clc = \lita_1 \lor \formuladots \lor \lita_{\clwidth}$
is a disjunction of literals over pairwise disjoint variables
(without loss of generality we assume that there are no trivial
clauses containing both a variable and its negation).
A clause $\clc'$ \introduceterm{subsumes} another clause~$\clc$ 
if every literal from $\clc'$ also appears in $\clc$.
A \introduceterm{$\clwidth$\nobreakdash-clause}
is a clause that contains at
most $\clwidth$~literals. A \introduceterm{CNF formula} 
$\fstd = \clc_1 \land \formuladots \land \clc_m$ 
is a conjunction of 
clauses, and $\fstd$ is a \introduceterm{\mbox{\kcnfform{}}} if it
consists of \xclause{\clwidth}{}s.
We write $\variables(\fstd)$ to denote the set of variables appearing
in a formula~$\fstd$.  We think of clauses and CNF formulas as sets:
the order of elements is irrelevant and there are no repetitions.

A \introduceterm{resolution refutation}
$\refof{\proofstd}{F}$
of an unsatisfiable CNF formula~$F$, which can also be referred to  as a
\introduceterm{resolution proof}
for (the unsatisfiability of)~$F$, 
is an ordered sequence of clauses 
$\proofstd = (\cld_1, \dotsc, \cld_{\stoptime})$
such that 
\mbox{$\cld_{\stoptime} = \emptycl$}
is  the empty clause containing no literals,
and each clause~$\cld_i$,
$ i \in [\stoptime] = \set{1, \ldots, \stoptime}$,
is either one of the clauses in~$\fstd$
(an \introduceterm{axiom})
or is derived from clauses
$\cld_j,\cld_k$  in~\proofstd{} 
with $j, k < i$
by the 
\introduceterm{resolution rule}
\begin{equation}
  \label{eq:resolution-rule}
  \resrule{\clb \lor \varx}
          {\clc \lor \stdnot{\varx}}
          {\clb \lor \clc} 
          \eqperiod        
\end{equation}
For technical reasons, it will also be convenient to permit a
\introduceterm{weakening rule}
\begin{equation}
  \label{eq:weakening-rule}
  \weakrule{\clb}
           {\clb \lor \clc}          
\end{equation}
allowing to derive a strictly weaker clause from a clause already derived,
although this rule is not essential.

With every resolution proof~$\proofstd$
we can associate a DAG~$G_{\proofstd}$
by having a sequence of vertices~$v_i$ 
on a line in order of increasing~$i$, 
labelled by the clauses~$\cld_i \in \proofstd$,
and with directed edges 
$(v_j, v_i)$
and
$(v_k, v_i)$
if the clause~$\cld_i$ was derived by resolution from $\cld_j$
and~$\cld_k$
or an edge 
$(v_j, v_i)$
if $\cld_i$ was derived from~$\cld_j$ by weakening.
Note that there might be several occurrences of a
clause~$\cld$ in the proof~$\proofstd$, and if so each occurrence gets
its own vertex in~$G_{\proofstd}$.

\ifthenelse{\boolean{conferenceversion}}
{}
{Now we can formally define the proof complexity measures discussed in
  \refsec{sec:intro}.}
The
\introduceterm{length}
$\lengthofarg{\proofstd}$
of a resolution proof~$\proofstd$ is the number of clauses in it
(counted with repetitions).
The 
\introduceterm{width}
$\widthofarg{\clc}$
of a clause $\clc$ is $\setsizesmall{\clc}$,
\ie the number of literals,
\ifthenelse{\boolean{conferenceversion}}
{and $\widthofarg{\proofstd}$}
{and the width $\widthofarg{\proofstd}$ of a proof $\proofstd$}
is the size of a largest clause in~$\proofstd$. The
\emph{(clause) space}
at step~$i$
is the number of clauses $\clc_j$, $j<i$,
with edges to  clauses
$\clc_k$, $k \geq i$ in~$G_{\proofstd}$,
plus~$1$ for the clause~$\clc_i$ derived at this step.
Intuitively, space measures the number of clauses we need to keep in
memory at step~$i$, since they were derived before step~$i$ but are
used to infer new clauses at or after step~$i$.
The space~$\clspaceof{\proofstd}$ of a proof~$\proofstd$
is the maximum space over all steps in~$\proofstd$.
Taking the minimum over all 
\ifthenelse{\boolean{conferenceversion}}
{refutations,}
{resolution refutations of a CNF formula~$F$,}
we define the length, width, and space of refuting~$F$, respectively, as
$\lengthref{F} 
= \minofexpr[\refof{\proofstd}{F}]{\lengthofarg{\proofstd}}$,
$\widthref{F} 
= \minofexpr[\refof{\proofstd}{F}]{\widthofarg{\proofstd}}$,
and
$\clspaceref{F} 
= \minofexpr[\refof{\proofstd}{F}]{\clspaceof{\proofstd}}$.
We remark that any applications of the weakening
rule~\refeq{eq:weakening-rule} can always be eliminated from a
refutation without increasing the length, width, or space.

When reasoning about space, it is sometimes convenient to use a
slightly different, but equivalent, description of resolution that
makes explicit what clauses are in memory at each point in time.  We
say that a \introduceterm{configuration-style resolution refutation}
is a sequence
$(\clsd_0,\ldots, \clsd_\stoptime)$ 
of sets of clauses, or \introduceterm{configurations}, such that
$\clsd_0 = \emptyset$,
$\emptycl \in \clsd_\stoptime$,
and for all $t\in [\stoptime]$ the configuration $\clsd_t$
is obtained from $\clsd_{t-1}$ by one of the following
\introduceterm{derivation steps}: 
\begin{description}
\item[Axiom download] 
  $\clsd_t=\clsd_{t-1}\union\set{\clc}$, where $\clc$
  is a clause $\clc \in F$.

\item[Inference] 
  $\clsd_t=\clsd_{t-1} \union \set{\cld}$ for
  a clause $\cld$ derived by resolution~\refeq{eq:resolution-rule}
  or weakening~\refeq{eq:weakening-rule}
  from
  clauses in~$\clsd_{t-1}$.
  
\item[Erasure] 
  $\clsd_t =  \clsd_{t-1} \setminus \clsd'$ 
  for some $\clsd' \subseteq \clsd_{t-1}$.
  \end{description}
The length of a configuration-style refutation
$\proofstd = (\clsd_0,\ldots, \clsd_\stoptime)$ 
is the number of axiom downloads and inference steps,
the width is the size of a largest clause, as before,  
and the space is
$
\maxofexpr[{t \in [\stoptime]}]{\setsize{\clsd_t}}
$.
Given a refutation as an ordered sequence of clauses 
$\proofstd = (\cld_1, \dotsc, \cld_{\stoptime})$,
we can construct a configuration-style refutation in the same length,
width, and space by deriving each clause $\cld_i$ via an axiom
download or inference step, and interleave with erasures of
clauses~$\clc_j$, $j<i$, as soon as 
these clauses
have no edges to
clauses~$\clc_k$, $k \geq i$, in the associated DAG~$G_{\proofstd}$.
In the other direction, taking a configuration-style refutation 
and listing the sequence of axiom download and inference steps yields
a standard resolution refutation in the same length, width, and space 
(assuming that clauses are erased as soon as possible).
Thus, we can switch freely between these two ways of describing
resolution refutations.

In this paper,
it will be convenient for us to limit our attention to a
(slightly non-standard) restricted form of resolution refutations as
described next. 
We define a
\introduceterm{\homogeneous resolution  refutation} 
to be a 
refutation where every resolution rule application is of the form
\begin{equation}
  \label{eq:homogeneous-resolution-rule}
  \resrule{\clc \lor \varx}
  {\clc \lor \stdnot{\varx}}
  {\clc} 
  \eqperiod        
\end{equation}
The requirement of \homogeneity is essentially without loss of
generality, since we need to insert at most two weakening steps before
each application of the resolution rule, which increases the width by
at most~$1$, and the weakened clauses can then immediately be
forgotten. We state this observation formally for the record.

\begin{observation}
  \label{obs:homogeneous}
  If a CNF formula $F$ has a standard resolution refutation   without
  weakening steps in length~$\lengthstd$, width~$\widthstd$,
  and 
  space~$\clspacestd$,
  then it has a \homogeneous refutation 
  in length at most~$3\lengthstd$, width at most~$\widthstd + 1$,
  and 
  space at most~$\clspacestd + 2$.
\end{observation}

As already mentioned, 
a useful trick to obtain hard CNF formulas for different proof systems
and complexity measures, which will play a key role also in this
paper, is \introduceterm{\xorification{}}, \ie substituting variables
by exclusive ors of new variables and expanding out in the canonical
way to obtain a new CNF formula.  For example, the standard way to
define binary \xorsubstitution for a positive literal~$\varx$ is
\begin{equation}
  \substform{\varx}{\oplus_2}
  = 
  (\varx_{1} \lor \varx_{2})
  \land
  (\olnot{\varx}_{1} \lor \olnot{\varx}_{2})
  \eqcomma
\end{equation}
for a negative literal~$\olnot{\vary}$ we have
\begin{equation}
  \substform{\olnot{\vary}}{\oplus_2}
  = 
  (\vary_{1} \lor \olnot{\vary}_{2})
  \land
  (\olnot{\vary}_{1} \lor \vary_{2})
  \eqcomma
\end{equation}
and applying binary \xorsubstitution to the clause
$\varx \lor \olnot{\vary}$
we obtain the CNF formula
\begin{equation}
  \label{eq:example-subst-clause}
  \begin{split}
    \substform{(\varx \lor \olnot{\vary})}{\oplus_2}
    =
    \substform{\varx}{\oplus_2} \lor \substform{\olnot{\vary}}{\oplus_2}
    = \ \ \
    &(\varx_{1} \lor \varx_{2} \lor \vary_{1} \lor \olnot{\vary}_{2})
    \land
    (\varx_{1} \lor \varx_{2} \lor \olnot{\vary}_{1} \lor \vary_{2})
    \\
    \land \
    &(\olnot{\varx}_{1} \lor \olnot{\varx}_{2} \lor   \vary_{1} \lor \olnot{\vary}_{2})
    \land
    (\olnot{\varx}_{1} \lor \olnot{\varx}_{2} \lor \olnot{\vary}_{1} \lor \vary_{2})
    \eqperiod
  \end{split}
\end{equation}
The \xorification of a CNF formula~$F$ is the conjunction of all the
formulas corresponding to the \mbox{\xorified{}} clauses of~$F$.
We trust that the reader has no problems parsing this slightly
informal definition by example or generalising it to substitutions
with \XOR of arbitrary arity 
(but see, e.g., Definition~2.12 in~\cite{Nordstrom13SurveyLMCS} for a
more rigorous treatment).

Usually, \xorification is done in such a way that any two variables in
the original formula are replaced by exclusive ors over disjoint sets
of new variables.  Razborov~\cite{Razborov16NewKind} observed that it
can sometimes be useful to allow \xorification with overlapping sets
of variables. Let us define this concept more carefully.

\begin{definition}[\xorification with recycling \cite{Razborov16NewKind}]
  \label{def:xor-substitution}
  Let $F$ be a CNF formula over
  the set of
  variables
  $u_1, \ldots, u_{\orignvars}$
  and let 
  $\expandergraph = (\leftvertexset \disjointunion \rightvertexset,E)$ 
  be a bipartite graph
  with left vertex set
  $\leftvertexset = \set{u_1, \ldots, u_{\orignvars}}$
  and right vertex set
  $\rightvertexset = \set{v_1, \ldots, v_{\newnvars}}$.
  Then for the variables $u_i$ we define the \xorified literals
  \mbox{$\substform{u_i}{\expandergraph} = \bigoplus_{v \in \nbhd(u_i)} v$}
  and
  \mbox{$\substform{\olnot{u}_i}{\expandergraph} = \lnot\bigoplus_{v \in \nbhd(u_i)} v$}
  (where
  $\nbhd(u_i)$
  denotes the neighbours in~$\rightvertexset$ of~$u_i$),
  for clauses $\clc \in F$ we define
  $\substform{\clc}{\expandergraph}
  = \Lor_{\lita \in \clc} \substform{\lita}{\expandergraph}$
  expanded out to CNF 
  as in~\refeq{eq:example-subst-clause}
  but with trivial clauses pruned away, 
  and the 
  \introduceterm{\xorification of~$F$ \wrt~$\expandergraph$} is    
  defined to be
  $\substform{F}{\expandergraph} =
  \Land_{\clc \in F} \substform{\clc}{\expandergraph}$.
\end{definition}

Note that if $F$ is an $\orignvars$-variable $\clwidth$-CNF 
with~$m$ clauses
and 
$\expandergraph = (\set{u_1, \ldots, u_{\orignvars}}
\disjointunion \set{v_1, \ldots, v_{\newnvars}},E)$    
is a bipartite graph of left degree $\expanderdegree$, then
$\substform{F}{\expandergraph}$ is an $\newnvars$-variable
$\clwidth\expanderdegree$-CNF formula 
with most
$2^{\expanderdegree-1} m$ clauses. 
We want to highlight that by definition we have the equality
\begin{equation}
  \label{eq:substitution-and-or}
  \substform{(\clc \lor \lita)}{\expandergraph}  
  =
  \substform{\clc}{\expandergraph}  
  \lor
  \substform{\lita}{\expandergraph}  
\end{equation}
(where we can view the expressions in  
\refeq{eq:substitution-and-or}
either as the Boolean functions computed by these formulas
or as the corresponding clause sets but with trivial clauses removed),
and this will be convenient to use in some of our technical arguments.

We conclude this section with two simple observations that will also be
useful in what follows.

\begin{observation}
  \label{obs:width-xorified}
  If $F$ has a
  (homogeneous)
  resolution refutation in width~$\widthstd$ and
  $\expandergraph$ has left degree bounded by~$\expanderdegree$, then
  $\substform{F}{\expandergraph}$ can be refuted in 
  (homogeneous)
  resolution in
  width~$2\expanderdegree\widthstd$.
\end{observation}

This is not hard to show, and follows, e.g., 
\ifthenelse{\boolean{conferenceversion}}
{from} 
{from the proof of}
Theorem~2 in \cite{BN11UnderstandingSpace} (strictly speaking, this
theorem is for \xorification \emph{without} recycling, but recycling
can only decrease the width).

\begin{observation}
  \label{obs:depthlargerthanclausespace}
  If $F$ has a 
  (homogeneous)
  resolution refutation~$\proofstd$ such that the associated
  DAG~$G_{\proofstd}$ has depth (\ie longest path)~$\clspacestd$,
  then $\proofstd$ can be carried out
  (in homogeneous resolution)
  in space
  $\clspacestd + 2$
  (possibly by repeating and/or reordering clauses in~$\proofstd$).
\end{observation}

This second observation is essentially due to~\cite{ET01SpaceBounds}.
To see why this is true, note that the
proof DAG~$G_{\proofstd}$ can be turned into a binary tree of the same
depth by repeating vertices/clauses, and it is then straightforward to
show that any tree-like proof DAG in depth~$\clspacestd$ can be
realized in space at most $\clspacestd + 2$. 

\makeatletter{}%

\section{Proof of Main Theorem}
\label{sec:proof-main-thm}

In this section we present a proof of \refthm{thm:maintheorem}.  The
proof makes use of the following hardness condensation lemma, which
will be established in the next section
and is the main technical contribution of the paper.

\begin{lemma}[Hardness condensation lemma]
  \label{lem:HardnessCondensingSpace}
  For all $\initialcnfwidth \in \Nplus$ and $\smallepsilon>0$ 
  there exist  $\newnvarszero \in \Nplus$ and~$\smalldelta>0$
  such that the following holds.
  Let $\widthlower$ and $\newnvars$ 
  be integers satisfying 
  $\newnvars\geq\newnvarszero$
  and
  $\initialcnfwidth\leq\widthlower \leq
  \newnvars^{\frac{1}{2}-\smallepsilon}$,
  and  
  suppose that $\formf$ is an unsatisfiable
  CNF formula
  over $\orignvars = \lfloor\newnvars^{\delta\widthlower}\rfloor$
  variables   which requires width
  $\widthref{\formf} = \initialcnfwidth$
  and
  space
  $\clspaceref{\formf} = \clspacesym$
  to be refuted in resolution.
 
  Then there is a bipartite graph
  $\graph=(\leftvertexset\disjointunion\rightvertexset,E)$ 
  with
  $\setsize{\leftvertexset}=\orignvars$
  and
  $\setsize{\rightvertexset}=\newnvars$
  such that the $\newnvars$-variable  
  CNF formula $\formf\substituted$ has the following properties:
  \begin{itemize}
  \item 
    $\formf\substituted$ can be refuted in width $\widthlower$.
  \item     
    Any refutation
    $\refof{\proofstd}{\formf\substituted}$ in
    width $\widthstd\leq \widthlower \log\newnvars$ 
    requires space 
    $\clspaceof{\proofstd} \geq 
    {(\clspacesym-\widthstd-3)}2^{-\widthstd}$.
  \end{itemize}
\end{lemma}

We want to apply this lemma to formulas of low width complexity but
high space complexity as stated next.

\begin{theorem}[\cite{BN08ShortProofs}]
  \label{thm:formulas_small_width_large_clspace}
  There is a family $\set{\formf_\orignvars}_{\orignvars\in \mathbb N}$ of $\orignvars$-variable \xcnf{6} formulas
  of size~$\bigtheta{\orignvars}$
  which can be refuted in 
  width $\widthref{\formf_\orignvars} = 6$
  but require space
  $\clspaceref{\formf_\orignvars} = \bigomega{\orignvars/\log \orignvars}$.
\end{theorem}

Combining
\reflem{lem:HardnessCondensingSpace}
and
\refthm{thm:formulas_small_width_large_clspace}, 
we can prove our main result.

\begin{proof}[Proof of \refthm{thm:maintheorem}]
  Recall that we want to prove that for any constant $\varepsilon>0$
  and any non-decreasing function 
  $\widthlower(\indexnalt)
  \leq \indexnalt^{\frac{1}{2}-\epsilon}$ 
  there is a
  family $\set{\formf_\indexnalt}_{\indexnalt\in \mathbb N}$ of
  $\indexnalt$-variable \CNF{} formulas which 
  have resolution refutations 
  of width~$\widthlower(\indexnalt)$ 
  but for which any
  refutation of width  
  $\littleoh{\widthlower(\indexnalt)\log \indexnalt}$ 
  requires clause space 
  $\indexnalt^{\bigomega{\widthlower(\indexnalt)}}$.

  From \refth{thm:formulas_small_width_large_clspace} we obtain constants
  $\smallepsilonalt>0$ and~$\orignvarszero \in \Nplus$ and a family of
  \mbox{$\orignvars$-variable}  \xcnfform{6}{}s $\formf_\orignvars$  
  that require clause space $\smallepsilonalt\orignvars/\log
  \orignvars$ for all $\orignvars\geq\orignvarszero$.  
  \ifthenelse{\boolean{conferenceversion}}
  {We want to apply \reflem{lem:HardnessCondensingSpace}}
  {We want to apply hardness
    condensation as in \reflem{lem:HardnessCondensingSpace}}
  to these formulas. 
  Let $\smallepsilon>0$ be given in Theorem~\ref{thm:maintheorem} and fix
  $\initialcnfwidth=6$.
  Plugging this into 
  \reflem{lem:HardnessCondensingSpace}
  provides constants
  $\smalldelta>0$ and $\newnvarszero \in \Nplus$,
  where in addition we choose $\newnvarszero$
  large enough so that
  $\lfloor\newnvarszero^{\delta\widthlower(\newnvarszero)}\rfloor\geq
  \orignvarszero$  
  (this is always possible since
  $\delta\widthlower(\newnvarszero) \geq 6\delta > 0$).

  For any 
  $\newnvars\geq\newnvarszero$, 
  set
  $\orignvars = 
  \lfloor\newnvars^{\delta\widthlower(\newnvars)}\rfloor\geq
  \orignvarszero$ 
  and let
  $\graph=(\leftvertexset\disjointunion\rightvertexset,E)$ 
  with
  $\setsize{\leftvertexset}=\orignvars$
  and
  $\setsize{\rightvertexset}=\newnvars$
  be a bipartite graph with properties as guaranteed by
  \reflem{lem:HardnessCondensingSpace}.
  Then the lemma says that
  $\formf_\orignvars\substituted$ is an
  \mbox{$\newnvars$-variable} formula 
  which can be refuted in width $\widthlower$, but 
  for for which every refutation of width
  \mbox{$\widthstd\leq \frac{\widthlower}{4\initialcnfwidth} \log\newnvars$} 
  requires clause space
  $ {(\clspacesym-\widthstd-3)}2^{-\widthstd}$, where
  $\spacelowerbound \geq \smallepsilonalt{\orignvars/\log \orignvars}
  = 
  \smallepsilonalt{\lfloor\newnvars^{\delta\widthlower(\newnvars)}\rfloor/\log
    \lfloor\newnvars^{\delta\widthlower(\newnvars)}\rfloor}$ 
  is the space lower  bound for~$\formf_\orignvars$.   Choosing
  $\widthstd
  \leq
  \frac{\smalldelta}{2}\cdot
  \widthlower(\newnvars) \log\newnvars$ 
  (recall that
  $\widthstd      =
  \littleoh{\widthlower(\newnvars)\log\newnvars}$ 
  by assumption), the sequence of
  calculations 
  \begin{equation}
    \label{eq:final-equation-main-thm}
    {(\clspacesym-\widthstd-3)}2^{-\widthstd}
    \geq
    \bigl(
    \smallepsilonalt{\lfloor\newnvars^{\delta\widthlower(\newnvars)}\rfloor /
      \log \lfloor\newnvars^{\delta\widthlower(\newnvars)}\rfloor} -
    \tfrac{\smalldelta}{2}\widthlower(\newnvars) \log\newnvars
    \bigr)
    2^{- \tfrac{\smalldelta}{2}\widthlower(\newnvars) \log\newnvars}  
    \geq
    \BIGOMEGA{\newnvars^{\frac{\smalldelta}{2}\widthlower(\newnvars)}} 
  \end{equation}  
  yields the desired space lower bound.
\end{proof}

If one looks more closely at what is going on inside the proof of 
\refthm{thm:maintheorem},
where
\reflem{lem:HardnessCondensingSpace}
and
\refthm{thm:formulas_small_width_large_clspace}
come together, 
one can make
a 
somewhat intriguing observation.

As discussed in the introduction,
\refthm{thm:formulas_small_width_large_clspace}
is shown by using so-called pebbling formulas, which we now describe briefly.
Given a DAG~$\graphg$ with sources~$S$ and a unique sink~$z$, and with
all non-sources having fan-in~$2$, we let every vertex in~$\graphg$
correspond to a variable and define the 
\ifthenelse{\boolean{conferenceversion}}
{\introduceterm{pebbling formula} $\pebcontr[\graphg]{}$ to consist of the following
  clauses:}  
{\introduceterm{pebbling formula}  over~$\graphg$,
  denoted~$\pebcontr[\graphg]{}$, to consist of the following
  clauses:}  
\begin{itemize}
  
\item
  for all
  $s \in S$,
  the  clause
  $s$;

\item
  For all non-source vertices $v$ with
  predecessors
  $u_1, u_2$,    the clause
  $\olnot{u}_1 \lor \olnot{u}_2 \lor v$;

\item
  for the sink $z$,
  the 
  clause $\olnot{z}$.

\end{itemize}
Applying standard binary \xorsubstitution (without recycling)
as in~\refeq{eq:example-subst-clause}
to these formulas amplifies 
lower bounds on the number of
variables in memory  
$\varspaceref{\pebcontr[\graphg]{}}$ 
(which follow from properties of the chosen DAG~$\graphg$)
to 
lower bounds on the number of clauses 
$\clspaceref{\substform{\pebcontr[\graphg]{}}{\oplus_2}}$. 
In 
\reflem{lem:HardnessCondensingSpace}
we then do another round of \xorsubstitution, this time with recycling, to
decrease the number of variables while
maintaining the space lower bound for small-width refutations. 
It is not entirely clear why we would need two separate rounds of
\xorification  to achieve this result. In one sense, it would seem
more satisfying to get a
clean one-shot argument that just takes pebbling formulas and yields
the supercritical trade-offs by only one round of \xorification.

And in fact, if we are willing to accept a slightly weaker bound, we could
make such a one-shot argument and apply substitution with recycling directly 
to the pebbling formulas. The reason for this is that one can
actually prove a somewhat stronger version of hardness condensation
than in \reflem{lem:HardnessCondensingSpace}, as we will see
in \refsec{sec:hardness-condensation}. There is no need to
require that the original formula should have high space complexity
unconditionally, but it suffices that the formula exhibits a strong
trade-off between width and clause space. Since the number of clauses
times the maximal width of any clause is an upper bound on the total
number of distinct variables in memory, for any resolution
refutation~$\proofstd$ we have the inequality
$\clspaceof{\refutationstd} \cdot \widthofarg{\refutationstd} 
\geq \varspaceof{\refutationstd}$.
In~\cite{Ben-Sasson09SizeSpaceTradeoffs} a variable space lower
bound 
$\varspaceref{\pebcontr[\graphg]{}} =
\bigomega{\orignvars / \log  \orignvars}$
was presented (for appropriately chosen DAGs~$\graphg$), 
implying that any \mbox{width-$\widthstd$} refutation 
requires clause space at least
$\bigomega{\orignvars/(\widthstd \log \orignvars)}$.  
Since our hardness condensation step incurs a loss of a
factor~$1/2^\widthstd$, by starting with standard pebbling formulas
and applying \xorification with recycling directly we could obtain
asymptotically similar bounds to those in \refthm{thm:maintheorem}
in one shot.

However, one can also argue that by combining
\reflem{lem:HardnessCondensingSpace}
and
\refthm{thm:formulas_small_width_large_clspace}
in the way done above
one obtains a more modular proof, which shows that 
any formulas satisfying the conditions in 
\refthm{thm:formulas_small_width_large_clspace}
can be used for hardness condensation in a black-box fashion.
This is why we chose to present the proof in this way.
\makeatletter{}%

\section{Hardness Condensation}
\label{sec:hardness-condensation}

Let us now prove the hardness condensation lemma. We 
establish
a slightly stronger version of the lemma below, which clearly subsumes
\reflem{lem:HardnessCondensingSpace}.

\begin{lemma}[Hardness condensation lemma, strong version]
  \label{lem:StrongHardnessCondensingSpace} 
  For all 
  $\initialcnfwidth \in \Nplus$ 
  and $\smallepsilon>0$
  there are  
  $\newnvarszero \in \Nplus$
  and
  $\smalldelta>0$
  such that the following holds.
  Let $\widthlower$ and $\newnvars$ be integers satisfying 
  $\newnvars\geq\newnvarszero$ 
  and
  $\initialcnfwidth\leq\widthlower \leq \newnvars^{\frac{1}{2}-\smallepsilon}$ 
  and suppose that $\formf$ is an unsatisfiable
  CNF formula  
  over $\orignvars = \lfloor\newnvars^{\delta\widthlower}\rfloor$ variables 
  which requires width
  $\widthref{\formf} = \initialcnfwidth$
  to be refuted in resolution.

  Then there is a bipartite graph
  $\graph=(\leftvertexset\disjointunion\rightvertexset,E)$ 
  with
  $\setsize{\leftvertexset}=\orignvars$
  and
  $\setsize{\rightvertexset}=\newnvars$
  such that the $\newnvars$-variable  
  CNF formula $\formf\substituted$ has the following properties:
  \begin{itemize}
  \item 
    The \xorified formula~$\formf\substituted$ 
    can be refuted in width $\widthlower$.
  \item     
    \ifthenelse{\boolean{conferenceversion}}
    {Any refutation}
    {Any resolution refutation}
    $\refof{\proofstd}{\formf\substituted}$ 
    of the \xorified formula~$\formf\substituted$ 
    in width $\widthstd\leq \widthlower \log\newnvars$ 
    requires space 
    $\clspaceof{\proofstd} \geq 
    {(\clspacesym-\widthstd-3)}2^{-\widthstd}$,
        where 
    $\clspacesym$ is the 
    minimal space of any refutation 
    $\refof{\proofstd'}{\formf}$
    of the original formula~$\formf$ 
    in width at most
    $\widthstd$.
  \end{itemize}
\end{lemma}

Clearly, the key to obtain \reflem{lem:StrongHardnessCondensingSpace}
is to choose the right kind of graphs.  As
in~\cite{Razborov16NewKind}, we use boundary expander graphs where
the right-hand side is significantly smaller than the left-hand side.
Let us start by giving a proper definition of these graphs and reviewing the 
properties that we need from them.  
Most of our discussion of
boundary expanders can be recovered from~\cite{Razborov16NewKind},
but since our setting of parameters is slightly different 
\ifthenelse{\boolean{conferenceversion}}
{we give a self-contained presentation below. We refer to the
  full-length version of this paper for any missing proofs.}
{we give a self-contained presentation and also provide full
  proofs of all claims in \refapp{app:appendix} for completeness.
  We remark that there is also a significant overlap
  with~\cite{BN16QuantifierDepth} in our treatment of expander graphs below.}

In what follows, we will let
$\expandergraph =
(\leftvertexset \disjointunion  \rightvertexset,E)$ 
denote a bipartite graph with left vertices~$\leftvertexset$
and
right vertices~$\rightvertexset$.
We write
$
\nbhd^\graph \bigl( \leftvertexsubset \bigr) 
=
\Setdescr{v}{\{u,v\}\in E(\graph),u\in \leftvertexsubset}$
to denote the set of right neighbours of a left vertex
subset~$\leftvertexsubset \subseteq \leftvertexset$
(and vice versa for right vertex subsets),
dropping the graph~$\graph$ from the notation when it is clear from
context. For a single vertex~$v$ we will use the 
abbreviation $\nbhd ( v ) = \nbhd ( \set{v} ) $.

\begin{definition}[Boundary expander]
  A bipartite graph 
  $\expandergraph = (\leftvertexset \disjointunion
  \rightvertexset,E)$ 
  is an
  \introduceterm{\nmboundaryexpnodegstd{}{}},
  or \introduceterm{unique neighbour expander},
  if
  \mbox{$\setsize{\leftvertexset}=\leftsize$},
  \mbox{$\setsize{\rightvertexset}=\rightsize$},
  and for every set 
  $\leftvertexsubset\subseteq\leftvertexset$,
  $\setsize{\leftvertexsubset}\leq \expansionguarantee$,
  it holds that
  $\setsize{\boundary(\leftvertexsubset)} \geq
  \expansionfactor\setsize{\leftvertexsubset}$, 
  where
  $\boundary(\leftvertexsubset) = 
  \Setdescr[:]{v \in \nbhd^{\graph}(\leftvertexsubset)}
  {\Setsize{\nbhd^{\graph}(v)\cap\leftvertexsubset} = 1}$ 
  is the \introduceterm{boundary} 
  or the set of \introduceterm{unique neighbours}
  of~$\leftvertexsubset$.
  \Aboundaryexpstd is  \aboundaryexpnodegstd where additionally 
  $\Setsize{\nbhd^{\graph}(u)} \leq \expanderdegree$ 
  for all
  $u\in\leftvertexset$,
  \ie where the left degree is bounded by~$\expanderdegree$. 
\end{definition}

\ifthenelse{\boolean{conferenceversion}}
{}
{An important property of \boundaryexpnodegstd{}s,
  which holds for arbitrarily small but positive expansion
  $\expansionfactor > 0$, is that  any left vertex subset
  $\leftvertexsubset \subseteq \leftvertexset$
  of size 
  $\setsize{\leftvertexsubset} \leq \expansionguarantee$
  has a  matching into~$\rightvertexset$.
  In addition, this matching can be chosen in such a way that there is
  an ordering of the  vertices in~$\leftvertexsubset$ such that every
  vertex $u_i \in \leftvertexsubset$ is matched to a vertex 
  outside of the neighbourhood of
  the preceding vertices $u_1, \ldots, u_{i-1}$.
  The proof of this fact uses what is sometimes
  referred to as a \introduceterm{peeling argument}, which we
  recapitulate below for the convenience of the reader.

  \begin{lemma}[Peeling lemma] \label{lem:peeling_lemma}
    Let
    $\expandergraph = (\leftvertexset \disjointunion \rightvertexset,E)$ 
    be   \aboundaryexpnodegstd with
    $\expansionguarantee\geq1$
    and
    $\expansionfactor > 0$.
    Then every left vertex subset 
    $\leftvertexsubset \subseteq \leftvertexset$
    of size 
    $\setsize{\leftvertexsubset} = \ell \leq \expansionguarantee$
    can be ordered
    $\leftvertexsubset = (u_1,\ldots,u_\ell)$
    in such a way that there is a matching into an ordered right vertex
    subset 
    $\rightvertexsubset = (v_1,\ldots,v_\ell) \subseteq
    \rightvertexset$ 
    for which
    $v_i\in \nbhd(u_i)\setminus \nbhd(\{u_1,\ldots,u_{i-1}\})$. 
  \end{lemma}
  
  \begin{proof}
    The proof is by induction on $\ell$.
    The base case $\ell=1$ is immediate since
    $\expansionguarantee\geq1$
    and
    $\expansionfactor > 0$ 
    implies that no left vertex can be isolated.
    For the induction step, suppose the lemma holds 
    for~$\ell - 1$.
    To define the sequence $v_1,\ldots,v_\ell$ we first fix
    any $v_\ell \in \boundary(\leftvertexsubset)$, which exists
    because   
    $\Setsize{\boundary(\leftvertexsubset)}
    \geq \mbox{$\expansionfactor \setsize{\leftvertexsubset} > 0$}$.
    Since $v_\ell$ is in the boundary of~$\leftvertexsubset$
    there exists a unique~$u_\ell \in \leftvertexsubset$ such  that  
    $\setsize{\nbhd(v_\ell)\cap\leftvertexsubset}=\{u_\ell\}$.
    Thus, for this pair $(u_\ell, v_\ell)$ it holds that
    $v_\ell \in \nbhd(u_\ell) \setminus 
    \nbhd \bigl(\leftvertexsubset\setminus\{u_\ell\}\bigr)$.
    By the induction hypothesis we can now find sequences
    $u_1,\ldots,u_{\ell-1}$ and $v_1,\ldots,v_{\ell-1}$ 
    for~$\leftvertexsubset\setminus\{u_\ell\}$
    such that $v_i\in
    \nbhd(u_i)\setminus \nbhd(\{u_1,\ldots,u_{i-1}\})$,
    to which we can append
    $u_\ell$ and~$v_\ell$ at the end. The lemma follows by the induction
    principle.
  \end{proof}
}

For a right vertex subset
$\rightvertexsubset \subseteq \rightvertexset$ 
in
$\expandergraph = (\leftvertexset \disjointunion \rightvertexset,E)$
we define the \introduceterm{kernel} 
\mbox{$\Ker \bigl( \rightvertexsubset \bigr) \subseteq \leftvertexset$}
of~$\rightvertexsubset$
to be the set of all left vertices whose entire neighbourhood is
contained in 
$\rightvertexsubset$, \ie
\begin{equation}
  \label{eq:def-kernel}
  \Ker \bigl( \rightvertexsubset \bigr) =
  \Setdescr{\mbox{$u\in\leftvertexset$}}{\,\nbhd(u)\subseteq\rightvertexsubset}
  \eqperiod
\end{equation}
We write
$\expandersubgraph{\expandergraph}{\rightvertexsubset}$ 
to denote the subgraph of $\expandergraph$ induced on
$\bigl( \leftvertexset\setminus\Ker(\rightvertexsubset) \bigr)
\disjointunion
\bigl( \rightvertexset\setminus\rightvertexsubset \bigr)$.
In other words, we can think of 
$\expandersubgraph{\expandergraph}{\rightvertexsubset}$
as being obtained from~$\expandergraph$ by first deleting
$\rightvertexsubset$ and afterwards all isolated vertices
from~$\leftvertexset$.

\ifthenelse{\boolean{conferenceversion}}
{A key property} 
{Another key property} 
of boundary expanders is that
for any small enough right vertex
set~$\rightvertexsubset$ we 
can always find a \introduceterm{closure}
$\closure\bigl(\rightvertexsubset\bigr) \supseteq \rightvertexsubset$
with a small kernel on the left  such that the subgraph
$\expandersubgraph{\expandergraph}{\closure(\rightvertexsubset)}$ 
has good boundary expansion. This is very similar to an analogous lemma 
\ifthenelse{\boolean{conferenceversion}}
{in~\cite{Razborov16NewKind}. We omit the proof due to space
  constraints.} 
{in~\cite{Razborov16NewKind},
  but since our parameters are
  slightly different we provide a proof of the next lemma in
  \refapp{app:appendix}.
}

\begin{lemma}
  \label{lem:ClosedSet}
  Let $\expandergraph$ be an
  \boundaryexpnodeg{\expguarantee}{2}.
  Then for every $\rightvertexsubset \subseteq \rightvertexset$ with
  $\setsize{\rightvertexsubset}\leq \expansionguarantee/2$ 
  there exists a set of vertices
  $\closure(\rightvertexsubset) \supseteq \rightvertexsubset$
  such that 
  $\Setsize{\Ker \bigl( \closure \bigl( \rightvertexsubset \bigr)\bigr)}
  \leq \setsize{\rightvertexsubset}$
  and the induced subgraph
  $\expandersubgraph{\expandergraph}{\closure(\rightvertexsubset)}$ is
  an
  \boundaryexpnodeg{\expguarantee/2}{1}.
\end{lemma}

The  next lemma states that there exist
\nmboundaryexp{\leftsize}{\rightsize}{\expdegree}{\expguarantee}{2}{}s
where the size $\rightsize$ of the right-hand side is significantly
smaller than the size $\leftsize= n^{\bigtheta\expdegree}$ of the
left-hand side.   
\ifthenelse{\boolean{conferenceversion}}
{This can be proven by a standard application of the probabilistic method.}
{The proof, which closely follows  \cite[Lemma 2.2]{Razborov16NewKind},  is a
  standard application of the probabilistic method,  but is  included
  in \refapp{app:appendix} for completeness.}  

\newcommand{\newexpanderexistTEXT}{%
  Fix constants 
  $\smallepsilon, \smalldelta > 0$ and $\mindegree \geq 2$ 
  such that
  $\smalldelta + \frac{1}{\mindegree} < \smallepsilon/2$. 
  Then there exists an 
  $\rightsizezero\in\Nplus$ 
  such that for all 
  $\rightsize$,
  $\expanderdegree$,
  and
  $\expansionguarantee$ 
  satisfying 
  $\rightsize\geq\rightsizezero$,
  $\mindegree\leq\expanderdegree\leq \rightsize^{1/2 - \smallepsilon}$, 
  and
  $\expansionguarantee \leq \rightsize^{1/2}$ 
  there are
  \nmboundaryexp{\lfloor\rightsize^{\smalldelta\expanderdegree}\rfloor}{\rightsize}{\expanderdegree}{\expansionguarantee}{2}{}s.   
  }

\begin{lemma}%
  \label{lem:newexpanderexist}
  \newexpanderexistTEXT{}
\end{lemma}

After this review of boundary expanders and their properties
we now come to the core argument of the paper, 
namely that space lower bounds are preserved for small-width 
resolution refutations when we apply \XORification 
as in \refdef{def:xor-substitution} with  respect to an
\boundaryexpnodeg{\expansionguarantee}{2}. 
To get cleaner technical arguments in the proofs
we will restrict our attention to \homogeneous resolution refutations
as in~\refeq{eq:homogeneous-resolution-rule}, which for our purposes is
without loss of generality by~\refobs{obs:homogeneous}.

\begin{lemma}[Main technical lemma]
  \label{lem:HardnessCondensingLemma}
  Let $\formorig$ be 
  an unsatisfiable
  \CNF-formula and $\expandergraph$ an
  \boundaryexpnodeg{\expansionguarantee}{2},
  and suppose that 
  $\refof{\refofsubst}{\formsubst}$ 
  is a \homogeneous resolution refutation 
  in width 
  $\widthstd \leq \expansionguarantee/2$   
  of the \xorified formula~$\formsubst$.
  Then 
  there is a \homogeneous refutation $\refof{\refoforig}{\formorig}$ 
  of the original formula~$\formorig$
  in width at most~$\widthstd$ and space  
  $\clspaceof{\refoforig} \leq
  2^{\widthstd}\clspaceof{\refofsubst}+\widthstd+3$. 
\end{lemma}

\begin{proof}
  Assume that 
  $\refofsubst=(\clsc_0,\clsc_1,\ldots,\clsc_{\maxclscindex})$ 
  is a configuration-style \homogeneous resolution refutation
  of~$\formsubst$ in width  
  $\widthofarg{\refofsubst} = \widthstd \leq \expansionguarantee/2$. 
  We will show how to transform~$\refofsubst$
  into a refutation~$\refoforig$ of the original formula~$\formorig$ in
  width and space as claimed in the lemma.
  To help the reader navigate the proof, we remark that in what follows
  we will use the notational conventions that
  $\clother$   and~$\clc$ 
  denote clauses over $\variables(\formsubst)$,
  $\cltrans$ 
  denotes
  a clause
  over~$\variables(\formorig)$,
  and     $\claxiom$ 
  denotes
  an axiom clause 
  from
  the original formula~$\formorig$
  before \xorification.

  Recall that for clauses
  $\clc \in \formsubst$
  we have
  $\variables(\clc)\subseteq \rightvertexset$
  by construction. For convenience, we will
  overload notation and write $\Ker(\clc) = \Ker(\variables(\clc))$, 
  which is a subset of the variables~$\leftvertexset$ of the original
  formula~$\formorig$. 
  Furthermore, for every clause
  $\clc\in\refofsubst$
  we fix 
  $\closure(\clc) \defi
  \closure(\variables(\clc)) \subseteq \rightvertexset$ 
  to be a minimal closure with properties as guaranteed by
  \reflem{lem:ClosedSet},
  \ie such that 
  $\Setsize{\Ker \bigl( \closure \bigl( \rightvertexsubset \bigr)\bigr)}
  \leq \setsize{\rightvertexsubset}$
  and the induced subgraph
  $\expandersubgraph{\expandergraph}{\closure(\rightvertexsubset)}$ is
  an
  \boundaryexpnodeg{\expguarantee/2}{1}.
  Note that
  such closures exist since all clauses
  $\clc\in\refofsubst$ have width at most~$\widthstd$.
  It might be worth pointing out, though, that this is a purely
  existential statement---we have no control over how these closures
  are constructed, and, in particular,
  for two clauses
  $\clb$ and~$\clc$ such that
  $\clb \subseteq \clc$
  it does not necessarily hold that
  $\closure(\clb) \subseteq \closure(\clc)$.

  An important notion in what follows will be that of
  \introduceterm{\simfalsesubst{}}, where we say that two CNF formulas
  $\fstd$ and~$\falt$ are
  \introduceterm{\simfalse{}}
  if there is a truth value assignment that at the same time falsifies
  both $\fstd$ and~$\falt$.
  To transform the resolution refutation $\refofsubst$ of $\formsubst$ into a
  refutation $\refoforig$ of~$\formorig$ we let $\clsd_{\confindex}$ be
  obtained from~$\clsc_\confindex$ by  
  replacing every clause $\clc\in \clsc_\confindex$ by the 
  set of clauses
  \begin{align}
    \label{eq:def-F-clauses}
    \clauseset(\clc) 
    &\defi                   
      \setdescr
      {\cltrans\!}
      {\!\variables(\cltrans) = \Ker(\closure(\clc)) ;
      \simfalseargsmathtext{\cltranssubst}{\clc}}
    \\
    \shortintertext{and defining}
    \label{eq:def-D-clauses}
    \clsd_{\confindex} &\defi \textstyle\Union_{\clc\in\clsc_\confindex}       
                         \clauseset(\clc)
  \end{align}
  (where the notation 
  $\clauseset(\clc)$ 
  is chosen to suggest that this is in some intuitive sense the
  ``inverse  operation'' of \xorification \wrt~$\expandergraph$).
  Every clause in $\cltrans\in\clauseset(\clc)$ has width at
  most~$\widthstd$, because
  \begin{equation}
    \label{eq:vars-D-bound}
    \setsize{\variables(\cltrans)} 
    = 
    \setsize{\Ker(\closure(\clc))} 
    \leq 
    \widthofarg{\clc} 
    \leq 
    \widthstd    
    \eqcomma
  \end{equation}
  where the first inequality is guaranteed by \reflem{lem:ClosedSet}
  and the second inequality is by assumption.
  Furthermore, we have 
  $\setsize{\clauseset(\clc)}\leq 2^{\widthstd}$,
  since all clauses in
  $\clauseset(\clc)$
  are over the same set of variables 
  and each variable appears positively or negatively in every clause, 
  and hence 
  \begin{equation}
    \label{eq:size-clsd-bound}
    \Setsize{\clsd_{\confindex}}
    \leq
    2^{\widthstd} \Setsize{\clsc_{\confindex}}
    \leq
    2^{\widthstd} \clspaceof{\refofsubst}
    \eqperiod
  \end{equation}
  We want to argue that the sequence 
  $\bigl( \clsd_{0},\clsd_{1},\ldots,\clsd_{\maxclscindex} \bigr)$ 
  is the ``backbone'' of a resolution refutation~$\refoforig$ of~$\formorig$,
  by which we mean that for every $\confindex$ it holds that
  $\clsd_{\confindex+1}$ can be derived from~$\clsd_{\confindex}$ 
  by a sequence of intermediate steps without affecting any proof
  complexity measure too much. 

  To make this claim formal, 
  we first observe that for 
  $\clsc_{0} = \emptyset$
  we obviously get
  $\clsd_{0} = \emptyset$ 
  by~\refeq{eq:def-D-clauses}.
  Moreover, it holds that  
  $\clauseset(\emptycl)=\{\emptycl\}$ 
  and hence
  $\emptycl\in\clsd_{\maxclscindex}$,
  since the unique minimal closure of the empty set is the empty set itself.
  We want to show that for every
  \mbox{$0\leq\confindex<\maxclscindex$} 
  the configuration~$\clsd_{\confindex+1}$ can be obtained
  from~$\clsd_{\confindex}$ by a 
  resolution derivation
  \mbox{$(\clsd_{\confindex} = 
    \clsd_{\confindex,0},
    \clsd_{\confindex,1},
    \clsd_{\confindex,2},
    \ldots,
    \clsd_{\confindex, j_{\confindex} - 1},
    \clsd_{\confindex, j_{\confindex}}
    = \clsd_{\confindex+1})$,} 
  where the space of every intermediate configuration is bounded by
  $\maxofexpr{\clspaceof{\clsd_{\confindex}},
    \clspaceof{\clsd_{\confindex+1}}} + \widthstd + 3$. 

  If $\clsc_{\confindex+1}$ is obtained from $\clsc_{\confindex}$ by
  erasing a clause $\clc$, then $\clsd_{\confindex+1}$ can be obtained from
  $\clsd_{\confindex}$ by erasing all clauses
  $\clauseset(\clc)\setminus\clsd_{\confindex+1}$.  
  Suppose that $\clsc_{\confindex+1}$ is obtained from
$\clsc_{\confindex}$ by downloading an axiom $\clc\in \formsubst$.  
We
claim
that every clause in $\clauseset(\clc)$ is either an axiom or
a weakening of an axiom from $\formorig$.  
By the definition of $\formsubst$, every axiom $\clc\in \formsubst$ is
a clause in the CNF formula~$\claxiomsubst$ for some original axiom
$\claxiom\in\formorig$.  
Fix any axiom
$\claxiom\in\formorig$
such that
$\clc \in \claxiom\substituted$. 
Then for all
$\cltrans\in\clauseset(\clc)$ it holds
by~\refeq{eq:def-F-clauses} that
$\variables(\cltrans)=\Ker(\closure(\clc))\supseteq \Ker(\clc)\supseteq
\variables(\claxiom)$
and that there is an assignment falsifying both
$\cltranssubst$ and~$\clc$.
To see that this implies that 
$\claxiom$ subsumes~$\cltrans$, suppose that there is a
variable~$\varx$
appearing positively in $\claxiom$ such that
$\olnot{\varx} \in \cltrans$.
Any truth value assignment falsifying
$\cltranssubst$ 
must falsify 
$\lita\substituted$
for all literals $\lita \in \cltrans$, and hence in particular
$\olnot{\varx}\substituted$.
This means that 
${\varx}\substituted$
is satisfied by the same assignment,
and then so is all of the formula
$\claxiom\substituted$
including~$\clc$. But this is a contradiction
to the \simfalsesubst
of~$\cltranssubst$  and~$\clc$, and so not only does it hold that
$\variables(\claxiom) \subseteq \variables(\cltrans)$
but $\claxiom$ is in fact a subclause of~$\cltrans$ as claimed.
From this we see that  we can add the clauses $\clauseset(\clc)$ to
$\clsd_{\confindex}$ using axiom download and weakening.  
After applying a weakening step we immediately delete the old clause. 
Hence, the additional weakening might increase the space by at most one.
It follows that the space of the intermediate configurations need
never exceed 
$\clspaceof{\clsd_{\confindex+1}}+1$.   

It remains to 
argue
that $\clsd_{\confindex+1}$ can be derived
from~$\clsd_{\confindex}$ when $\clsc_{\confindex+1}$ is obtained
from~$\clsc_{\confindex}$ by an inference step.
\ifthenelse{\boolean{conferenceversion}}
{This is stated in the following two claims regarding applications of
  the resolution and weakening rules. Here graph expansion comes
  heavily into play, but due to space constraints we have to defer the
  proofs to the full-length version of this paper.} 
{This is stated in the following two claims regarding applications of
  the resolution and weakening rules.}

\begin{claim}
  \label{claim:resolutionstep} 
  Every clause $\cltrans\in \clauseset(\clc)$ can be derived from
  $\clauseset(\clc\lor \varx)\cup \clauseset(\clc\lor \stdnot{\varx})$
  by a \homogeneous resolution derivation of width $\widthstd$ and depth
  $\widthstd+1$. 
\end{claim}

\begin{claim}
  \label{claim:weakening}
  For any two clauses 
  $\clother$ and~$\clc$
  with 
  $\clother \subseteq \clc$ 
  it holds that
  every clause $\cltrans\in \clauseset(\clc)$ can be derived
  from~$\clauseset(\clother)$  by a \homogeneous
  derivation of width $\widthstd$ and depth $\widthstd+1$.   
\end{claim}

\ifthenelse{\boolean{conferenceversion}}
{}
{Taking these two claims on faith for now, let us see how they allow
  us  to  conclude the proof   of the lemma.}
Since
the depth of a refutation 
provides
an upper bound on the clause
space by \refobs{obs:depthlargerthanclausespace}, it follows
that in both cases we can derive 
all clauses in the clause set~$\clauseset(\clc)$  one by one
by using additional space $\widthstd+3$ to perform
the derivations in  depth~$\widthstd+1$.  
This shows 
that $\formorig$ has a \homogeneous resolution
refutation~$\refoforig$ of width $\widthstd$ and clause space 
$\clspaceof{\refoforig} \leq
2^{\widthstd}\clspaceof{\refofsubst}+\widthstd+3$,
which establishes the lemma.
\end{proof}

\newcommand{\graphclc}{\expandergraph_{\clc}}

\ifthenelse{\boolean{conferenceversion}}
{}
{%

  We proceed to establish
  Claims~\ref{claim:resolutionstep} and~\ref{claim:weakening}.

  \begin{proof}[Proof of Claim~\ref{claim:resolutionstep}]
    Recall that by \reflem{lem:ClosedSet}
    the subgraph
    $
    \graphclc
    \defi
    \expandersubgraph{\expandergraph}{\closure(\clc)}
    $ 
    is an
    \boundaryexpnodeg{\expansionguarantee/2}{1} and
    that for 
    $
    \Ker(\closure(\clc\lor \varx)) =
    \Ker(\closure(\clc\lor \olnot{\varx})) 
    $
    we have
    $\setsize{\Ker(\closure(\clc\lor \varx))} 
    \leq
    \widthofarg{\clc\lor \varx} 
    \leq
    \widthstd 
    \leq 
    \expansionguarantee/2$. 
    Therefore, we can apply 
    \reflem{lem:peeling_lemma} 
    to the set 
    $
    \auxkerset = 
    \Ker(\closure(\clc\lor \varx)) \setminus \Ker(\closure(\clc))$
    to obtain an ordering
    $u_1,\ldots,u_\ell$ 
    of~$\auxkerset$
    satisfying
    $\nbhd^{\graphclc} \! (u_i)\setminus
    \nbhd^{\graphclc} \! (\set{u_1,\ldots,u_{i-1}}) \neq \emptyset$. 
    For $0\leq i\leq\ell$ we let
  \begin{align}
    \auxkersetstar{i} 
    &\defi
      \big(\Ker(\closure(\clc))\cap\Ker(\closure(\clc\lor
      \varx))\big)\union \setdescr{u_j}{1\leq j \leq i}
    \\
    \intertext{%
    so that
    $\auxkersetstar{\ell} = \Ker(\closure(\clc \lor \varx))$
    and
    $\auxkersetstar{0} \subseteq \Ker(\closure(\clc))$, 
    and define}
    \label{eq:def-F-i}
    \clsetwithindex[i] 
    &\defi
      \setdescr{\cltrans}{\variables(\cltrans)=\auxkersetstar{i} ;
      {\simfalseargsmathtext{\cltranssubst}{\clc}}} 
      \eqperiod
  \end{align}
  Observe that
  \begin{align}
    \nonumber
    &
      \ \ \ \ \ \ \,
      \clauseset(\clc\lor\varx) \cup
      \clauseset(\clc\lor\stdnot{\varx}) 
    \\
    \nonumber
    &= 
      \ \ \ \ \ 
    \setdescr{\cltrans}{\variables(\cltrans)=\Ker(\closure(\clc\lor\varx));
      \text{\simfalseargs{\cltranssubst}{\clc\lor\varx}}} 
    \\ 
    &
      \ \ \ \ \, \,\,
      \cup 
    \setdescr{\cltrans}{
      \variables(\cltrans) = \Ker(\closure(\clc\lor\stdnot{\varx}));
    \text{\simfalseargs{\cltranssubst}{\clc\lor\stdnot{\varx}}}}
    \\ 
    \nonumber
    &=\; 
    \setdescr{\cltrans}
      {\variables(\cltrans)=\Ker(\closure(\clc\lor\varx));
      {\simfalseargsmathtext{\cltranssubst}{\clc}}} 
    \\ 
    \nonumber
    &=\; 
    \clsetwithindex[\ell]
  \end{align}
  and that every clause in $\clauseset(\clc)$ is   subsumed by
  a clause in $\clsetwithindex[0]$   since
  $
  \auxkersetstar{0} 
  \subseteq
  \Ker(\closure(\clc))
  $.
  Thus, we are done if we can derive all  clauses
  in~$\clsetwithindex[0]$ from the clauses in~$\clsetwithindex[\ell]$. 

  We
  do so inductively: for 
  $i=\ell,\ell-1, \ldots,2,1$ 
  we can obtain any clause
  $\cltrans\in\clsetwithindex[i-1]$ by an application of the 
  \homogeneous resolution rule to the clauses
  $\cltrans\lor u_{i}$ and $\cltrans\lor\olnot{u}_{i}$, 
  which we claim are both available in~$\clsetwithindex[i]$. 
  What remains to show is that $\cltrans\in\clsetwithindex[i-1]$
  indeed   implies  that 
  $
  \Set{\cltrans\lor u_{i}, \cltrans\lor\olnot{u}_{i}}
  \subseteq
  \clsetwithindex[i]$.  
  To argue this, note that
  by the definition of~$\clsetwithindex[i-1]$   in~\refeq{eq:def-F-i}
  there is a (partial) truth value assignment
  $\tvastd$ that simultaneously falsifies $\cltrans\substituted$ and~$\clc$. 
  The peeling lemma guarantees that 
  $
  \nbhd^{\graphclc} \! (u_i)\setminus
  \variables(\cltrans\substituted)
  =
  \nbhd^{\graphclc} \! (u_i)\setminus
  \nbhd^{\expandergraph}  \bigl( \auxkersetstar{i-1} \bigr)
  $ 
  has a non-empty intersection with
  $\rightvertexset\setminus\closure(C)$,
  the right-hand side of the expander~%
  $\graphclc$.
  Hence, we can extend $\tvastd$ and set the variables in
  $
  \nbhd^{\graphclc}\!(u_i) \setminus 
  \bigl( \variables(\cltrans\substituted) \cup \variables(C) \bigr)
  \supseteq
  \nbhd^{\graphclc}\! \bigl( \auxkersetstar{i} \bigr) 
  \setminus
  \nbhd^{\graphclc}\! \bigl( \auxkersetstar{i-1} \bigr) 
  \neq \emptyset
  $ 
  to 
  appropriate values
  so that the parity
  $\bigoplus_{v\in \nbhd(u_i)}\tvastd(v)$ 
  is even and thus 
  $(\cltrans\lor u_i)\substituted =
  \cltrans \substituted \lor u_i\substituted$ 
  is falsified,
  and we do so without assigning any variables in~$\clc$, which therefore
  remains falsified.
  In an analogous fashion, 
  by instead ensuring that the parity 
  $\bigoplus_{v\in \nbhd(u_i)}\tvastd(v)$ 
  is odd we get a
  falsifying assignment for 
  \mbox{$(\cltrans\lor\olnot{u}_i)\substituted \lor \clc$}.
  Hence, by~\refeq{eq:def-F-i} it holds that
  $\cltrans\lor u_{i}$ 
  and~$\cltrans\lor\olnot{u}_{i}$
  both  appear in~$\clsetwithindex[i]$. 

  Finally, to get from   $\clsetwithindex[0]$ to~$\clauseset(\clc)$ 
  we might need an extra weakening step as observed above. 
  The total depth of the whole derivation is at most $\ell+1 \leq \widthstd+1$. 
\end{proof}

\begin{proof}[Proof of Claim~\ref{claim:weakening}] 
  Note that if
  $\Ker(\closure(\clother)) \subseteq  \Ker(\closure(\clc))$ 
  this claim would be easy to establish, but as noted above we have no
  guarantee that  this is the case. Instead, we apply a proof strategy
  similar to the one for the previous claim.  
  We again have that
  $
  \graphclc \defi
  \expandersubgraph{\expandergraph}{\closure(\clc)}
  $ 
  is an \boundaryexpnodeg{\expansionguarantee/2}{1}, 
  so that we can apply the peeling lemma to the left-hand vertex set 
  $\Ker(\closure(\clother)) \setminus  \Ker(\closure(\clc))$ 
  to obtain an ordering $u_1,\ldots,u_\ell$ of
  its vertices satisfying 
  $\nbhd^{\graphclc} \! (u_i) \setminus
  \nbhd^{\graphclc} \! (\{u_1,\ldots,u_{i-1}\})
  \neq \emptyset$. 
  For $0\leq i\leq\ell$ we let 
  \begin{align}
     \auxkersetstar{i} 
    &\defi 
      \big(\Ker(\closure(\clc))\cap\Ker(\closure(\clother))\big)\cup
      \setdescr{u_j}{1\leq j \leq i} 
    \\
    \intertext{and as before define}
    \clsetwithindex[i] 
    &\defi
      \setdescr{\cltrans}{\variables(\cltrans)=\auxkersetstar{i} ;
      {\simfalseargsmathtext{\cltranssubst}{\clc}}} 
      \eqperiod
  \end{align}
  Note that $\clsetwithindex[\ell] \subseteq \clauseset(\clother)$, because if
  \simfalseargs{\cltranssubst}{\clc}, then
  \simfalseargscertainly{\cltranssubst}{\clother \subseteq \clc}. 
  Hence, we can obtain 
  $\clsetwithindex[\ell]$ from~$\clauseset(\clother)$ by just erasing clauses. 
  Once more, we apply the peeling argument
  in an inductive fashion
  and derive any 
  $\cltrans\in\clsetwithindex[i-1]$ from $\cltrans\lor u_{i}$
  and~$\cltrans\lor\olnot{u}_{i}$ 
  appearing
  in~$\clsetwithindex[i]$. In the end, we can
  infer
  any clause in~$\clauseset(\clc)$
  from~$\clsetwithindex[0]$ because every clause in~$\clauseset(\clc)$
  can be seen to be a weakening of some clause in~$\clsetwithindex[0]$.   
\end{proof}
}%

We can now combine the construction in
\reflem{lem:HardnessCondensingLemma} with the 
existence of good boundary expanders in \reflem{lem:newexpanderexist}
to prove the 
hardness condensation in \reflem{lem:StrongHardnessCondensingSpace}.

\begin{proof}[Proof of \reflem{lem:StrongHardnessCondensingSpace}]
  Given $\smallepsilon>0$ and  
  $\initialcnfwidth\in\Nplus$ we choose 
  $\smalldelta \defi
  \frac{\smallepsilon}{10\initialcnfwidth}$.  
  Note that we can assume
  $\smallepsilon \leq 1/2$ since otherwise the lemma is vacuous.
  Suppose
  $\widthlower$ and~$\newnvars$ 
  are parameters such that  
  $\initialcnfwidth\leq\widthlower \leq
  \newnvars^{\frac{1}{2}-\smallepsilon}$ 
  and 
  let $\formf$ be an unsatisfiable 
  CNF formula over
  $\orignvars = \lfloor\newnvars^{\delta\widthlower}\rfloor$
  variables that can be refuted in width $\initialcnfwidth$. 
  To apply \reflem{lem:newexpanderexist} we set
  $\mindegree \defi \frac{5}{\smallepsilon} > 2$ and  
  verify that
  $\smalldelta + \frac1\mindegree = 
  \frac{\smallepsilon}{10\initialcnfwidth} + \frac{\smallepsilon}{5} < 
  \frac{\smallepsilon}{2}$.
  We choose the degree of the expander to be 
  $\expanderdegree \defi
  \Floor{\frac{\widthlower}{2\initialcnfwidth}}$ 
  and set the size guarantee for expanding left vertex sets to 
  $\expansionguarantee \defi 2\widthlower\log \newnvars$. 
  By the bound on $\widthlower$ we have  $\expanderdegree \leq
  \widthlower \leq
  \newnvars^{\frac{1}{2}-\smallepsilon}$. Furthermore, we choose
  $\newnvarszero$ large enough so that
  $
  \expansionguarantee 
  \leq
  2 \newnvars^{\frac{1}{2}-\smallepsilon} \log \newnvars
  \leq 
  {\newnvars}^{\frac{1}{2}}
  $ 
  for 
  \mbox{all $\newnvars\geq\newnvarszero$.}

  Now we have two cases. 
  The first, and interesting, case is when
  $\expanderdegree \geq \mindegree$ holds. Then
  \reflem{lem:newexpanderexist} guarantees that there exists an  
  \nmboundaryexp{\orignvars}{\newnvars}{\expanderdegree}{\expansionguarantee}{2}
  $\expandergraph$.
  Applying \xorification \wrt~$\expandergraph$, we obtain 
  a CNF formula~$\formf\substituted$
  with $\indexnalt$~variables.  
  By \refobs{obs:width-xorified} it holds that
  $\formf\substituted$ has a resolution refutation of width
  $2\expanderdegree\initialcnfwidth\leq\widthlower$. 
  Now suppose that $\refof{\proofstd}{\formf\substituted}$ 
  is a refutation of width $\widthmax$.  Because 
  $\widthmax\leq \widthlower \log\newnvars 
  = \expansionguarantee / 2$ 
  the space lower bound 
  follows from
  \reflem{lem:HardnessCondensingLemma}.  

  The second case is when
  $\expanderdegree < \mindegree$. 
  Then we do not actually need any \xorification but can use the
  original formula. Formally,
  let
  $\expandergraph = 
  (\leftvertexset \disjointunion
  (\rightvertexset\cup\rightvertexset'),E)$ 
  be a matching between two sets $\leftvertexset$ and~$\rightvertexset$ 
  of size
  $\setsize{\leftvertexset}=\setsize{\rightvertexset}=\orignvars$ 
  plus some isolated vertices~$\rightvertexset'$ 
  on the right-hand side such that 
  $\setsize{\rightvertexset\cup\rightvertexset'} = \newnvars$. 
  To check that this is well defined we have to verify that
  $\orignvars\leq \newnvars$,
  which follows from the calculations
  $\orignvars =
  \lfloor\newnvars^{\delta\widthlower}\rfloor = 
  \lfloor\newnvars^{\frac{\smallepsilon}{10\initialcnfwidth}2\initialcnfwidth\expanderdegree}\rfloor 
  \leq 
  \lfloor\newnvars^{\frac{\smallepsilon}{10\initialcnfwidth}2\initialcnfwidth\mindegree}\rfloor
  = 
  \lfloor\newnvars^{\frac{\smallepsilon}{10\initialcnfwidth}2\initialcnfwidth\frac{5}{\smallepsilon}}\rfloor
  = 
  \newnvars
  $.
  In this somewhat convoluted way
  we obtain
  $\formf\substituted=\formf$ (plus some left-over dummy variables)
  and we have 
  $\mbox{$\widthref{\formf\substituted}$} = \widthref{\formf} =
  \initialcnfwidth\leq\widthlower$ as well as $\clspaceof{\proofstd}
  \geq \clspacesym \geq 
  {(\clspacesym-\widthstd-3)}2^{-\widthstd}$.
  The lemma follows.
\end{proof}

\makeatletter{}%
\section{Concluding Remarks}
\label{sec:conclusion}

In this paper we prove that there are CNF formulas over
$\nvars$~variables exhibiting 
an $n^{\bigomega{\widthstd}}$~clause space lower bound for resolution
refutations in width~$\widthstd$. 
This lower bound is optimal (up to constants in the exponent) as every
refutation in width~$\widthstd$ has length, and hence space, at
most $\nvars^{\bigoh{\widthstd}}$.  Our lower bounds do not only hold for
the minimal refutation width~$\widthstd$ but remain valid for
any refutations in 
width asymptotically smaller than $\widthstd\log \nvars$.
Measured in terms of the number of variables~$\nvars$, this is a major
improvement over the previous space-width trade-off result
in~\cite{Ben-Sasson09SizeSpaceTradeoffs}, and provides another example
of trade-offs in the supercritical regime above worst-case recently
identified in~\cite{Razborov16NewKind}.

Regarding possible future research directions, 
a first open problem is whether the range of applicability can be
extended even further so that the space lower bound holds true up to
width $\littleoh{n}$.  It is clear that the lower bound has to break
down at some point, 
since if one is allowed maximal width~$\nvars$ any formula can be
refuted in clause space~$\nvars + 2$ \cite{ET01SpaceBounds}.
A supercritical trade-off on resolution proof depth over width ranging from
$\widthstd$ all the way up to 
$n^{1-\epsilon}/\widthstd$ 
was  shown  in~\cite{Razborov16NewKind}, suggesting that the above
goal might not be completely out of reach.

Another intriguing open problem 
is to prove space trade-offs that are superlinear not only in
terms of the number of variables but measured 
also
in formula size.  Such
lower bounds cannot be obtained by the techniques used in this paper, 
but they are likely to exist as the
following argument shows (see \cite{AlexHertel08Thesis} for a more detailed
discussion).  Suppose that every refutation in width~$\widthstd(n)$
can be transformed into a refutation that has width~$\widthstd(n)$ and
clause space polynomial in the size of the formula.  
Then we can find such a refutation non-deterministically in polynomial space by keeping the current configuration in memory  and guessing the inference steps. 
Thus,
by Savitch’s theorem, finding refutations of width~$\widthstd(n)$ would be in deterministic \PSPACE. 
On the other hand, it has been shown by
\theauthorCB that the problem of finding resolution refutations of
bounded width is \EXPTIME-complete
\cite{Berkholz12ComplexityNarrowProofs}.  Hence, unless
$\EXPTIME=\PSPACE$ there are formulas
where every refutation of minimal width needs clause space that is
superpolynomial in the size of the formula.

Finally, it would be interesting to study if the supercritical
trade-offs between clause space and width in resolution shown in this
paper could be extended to similar trade-offs between monomial space
and degree for polynomial calculus or polynomial 
calculus
resolution
as defined in 
\cite{ABRW02SpaceComplexity,CEI96Groebner}.

\makeatletter{}%

\ifthenelse{\boolean{conferenceversion}}
{\section*{Acknowledgements}}
{\section*{Acknowledgements}}

We wish to thank Alexander~Razborov for patiently explaining
the  hardness condensation technique in~\cite{Razborov16NewKind}
during numerous and detailed discussions.

Part of the work of \theauthorCB was performed while at KTH Royal
Institute of Technology supported by a fellowship within the
Postdoc-Programme of the German Academic Exchange Service (DAAD).
The research of \theauthorJN was supported by the
European Research Council under the European Union's Seventh Framework
Programme \mbox{(FP7/2007--2013) /} ERC grant agreement no.~279611
and by
Swedish Research Council grants 
\mbox{621-2010-4797}
and
\mbox{621-2012-5645}.

\bibliography{refArticles,refBooks,refOther}

\bibliographystyle{alpha}

\appendix

\makeatletter{}%

\newcommand{\leftvertexsubsetsize}{\ell}
\newcommand{\eulernumber}{e}

\theoremstyle{plain}    

\newtheorem*{lem:expanderexists}{Lemma~\ref{lem:newexpanderexist} (restated)}
\newtheorem*{lem:closedset}{Lemma~\ref{lem:ClosedSet} (restated)}

\section{Appendix} 
\label{app:appendix}

In this appendix we give proofs for
\reftwolems{lem:ClosedSet}{lem:newexpanderexist}.
As already mentioned, most of this material appears in a similar form
in~\cite{Razborov16NewKind} (although the exact
parameters are slightly different), and there is also a substantial
overlap with analogous technical lemmas in~\cite{BN16QuantifierDepthECCC}.
In fact, \reflem{lem:ClosedSet} is exactly as stated in 
in~\cite{BN16QuantifierDepthECCC}, but we present a proof below to give
a self-contained exposition of our version of Razborov's hardness
condensation technique adapted to general resolution.

\begin{lem:closedset}
  Let $\expandergraph$ be a bipartite
  \boundaryexpnodeg{\expguarantee}{2}.
  Then for every right vertex set 
  $\rightvertexsubset \subseteq \rightvertexset$ 
  of size
  $\setsize{\rightvertexsubset}\leq \expansionguarantee/2$ 
  there exists a superset
  $\closure(\rightvertexsubset) \supseteq \rightvertexsubset$
  such that 
  $\Setsize{\Ker \bigl( \closure \bigl( \rightvertexsubset \bigr)\bigr)}
  \leq \Setsize{\rightvertexsubset}$
  and the induced subgraph
  $\expandersubgraph{\expandergraph}{\closure(\rightvertexsubset)}$ is
  an
  \boundaryexpnodeg{\expguarantee/2}{1}.
\end{lem:closedset}

\newcommand{\vset}{V}
\newcommand{\uset}{U}
\newcommand{\utilde}{\leftvertexsubset}
\newcommand{\ubar}{\overline{U}}
\renewcommand{\ubar}{U^*}
\newcommand{\finindex}{\tau}

\begin{proof}
  With assumptions as in the lemma,   let 
  $\expandergraph=(\leftvertexset \disjointunion \rightvertexset,E)$ 
  be an 
  \boundaryexpnodeg{\expguarantee}{2} and let
  $\rightvertexsubset\subseteq\rightvertexset$
  be a right vertex set of size
  \mbox{$\setsize{\rightvertexsubset}\leq \expansionguarantee/2$}. 
  We will construct an increasing sequence of right vertex sets
  $\rightvertexsubset=\vset_0\subset \vset_1 \subset \cdots \subset
  \vset_\finindex$ 
  such that for
  $\closure(\rightvertexsubset) = \vset_\finindex$
  it holds that
  $\expandersubgraph{\expandergraph}{\vset_\finindex}$ is an
  \boundaryexpnodeg{\expguarantee/2}{1}.

  If
  $\expandersubgraph{\expandergraph}{\vset_0}$ 
  is an
  \boundaryexpnodeg{\expguarantee/2}{1}, then we can stop right away,
  but otherwise there must exist a left vertex set
  $\uset_1$ of size 
  at most $\expguarantee/2$ 
  such that
  $\Setsize{\boundary^{\expandersubgraph{\expandergraph}{\vset_0}}(\uset_1)}
  \leq \setsize{\uset_1}$.  
  Delete
  $\uset_1$ and all its neighbours
  from~$\expandersubgraph{\expandergraph}{\vset_0}$.  
  If now the resulting graph is an
  \boundaryexpnodeg{\expguarantee/2}{1}, then we are done, 
  but otherwise we repeat this process and
  iteratively delete vertex sets that violate the expansion requirements.
  Formally, for $i\geq 1$ we let $\uset_i$ be any left vertex set of size 
  at most $\expansionguarantee/2$
  such that  
  $
  \Setsize{\boundary^{\expandersubgraph{\expandergraph}{\vset_{i-1}}}(\uset_i)}
  \leq
  \setsize{\uset_i}
  $,
  where we set 
  $
  \vset_{i} \defi \vset_0
  \cup
  \bigcup_{j=1}^{i} \nbhd^{\expandergraph}(\uset_j)
  $
  (and where we note that
  what is deleted at the
  \mbox{$i$th step} is 
  $\nbhd^{\expandersubgraph{\expandergraph}{\vset_{i-1}}}(\uset_i)$
  together with the  kernel~$
  \ker \bigl( 
  \nbhd^{\expandersubgraph{\expandergraph}{\vset_{i-1}}}(\uset_i)
  \bigl)
  $
  of this right vertex set,
  so that after the  \mbox{$i$th step} all of
  $\nbhd^{\expandergraph}(\uset_i)$ 
  and~$\ker(\nbhd^{\expandergraph}(\uset_i))$ has been removed from
  the graph).

  Since all sets~$\uset_i$ constructed above are non-empty,
  this process must terminate  for some
  $i=\finindex$ and the resulting graph
  $\expandersubgraph{\expandergraph}{\vset_\finindex}$ is then an
  \boundaryexpnodeg{\expguarantee/2}{1}
  (if nothing else, an empty graph without vertices vacuously satisfies the
  expansion  condition). However, we need to check that the condition
  $\setsize{\Ker(\vset_\finindex)}\leq \setsize{\vset_0}$
  holds.
  This follows from the next claim.
  
  \begin{claim}
    \label{claim:closure}
    Let 
    $\vset_{-1} = \uset_0 = \emptyset$ and suppose that $i\geq0$. 
    Then for $\uset_i$ and $\vset_i$ as constructed above
    we have the following properties:
    \begin{enumerate}
    \item
      \label{item:proofclosure1} 
      For all $\leftvertexsubset$ such that $\Ker(\vset_{i-1})\cup
      \uset_{i}\subseteq \leftvertexsubset \subseteq \Ker(\vset_i)$ 
      it holds that
      $\Setsize{\boundary^{\expandergraph}(\leftvertexsubset)\setminus\vset_0}
      \leq
      \setsize{\Ker(\vset_i)}$. 
    \item
      \label{item:proofclosure2} 
      The kernel of~$\vset_i$ has size
      $\setsize{\Ker(\vset_i)}\leq \setsize{\vset_0}$.
    \end{enumerate} 
  \end{claim}

\newcommand{\propclosureone}{Property~\ref{item:proofclosure1}\xspace} 
\newcommand{\propclosuredtwo}{Property~\ref{item:proofclosure2}\xspace} 
  \newcommand{\claimclosure}{Claim~\ref{claim:closure}\xspace}

  We establish Claim~\ref{claim:closure} by induction.
  For the base case $i=0$, 
  \propclosureone  
  holds since
  $\leftvertexsubset
  \subseteq \Ker(\vset_0)$ 
  implies  that
  $\boundary^{\expandergraph}(\leftvertexsubset)\subseteq \vset_0$.  
  For 
  \propclosuredtwo,
  suppose that $\setsize{\Ker(\vset_0)}\leq
  \expansionguarantee$.  
  Then by the expansion of~$\expandergraph$ we have that
  $
  2\setsize{\Ker(\vset_0)} \leq
  \setsize{\boundary^{\expandergraph}(\Ker(\vset_0))}$,
  and in combination with
  $\boundary^{\expandergraph}(\ker(\vset_0)) \subseteq \vset_0$
  this implies 
  $\setsize{\Ker(\vset_0)} \leq \frac{1}{2}\setsize{\vset_0}$.
  If instead
  $\setsize{\Ker(\vset_0)} > \expansionguarantee$,
  then we can find a subset
  $\leftvertexsubset\subseteq \Ker(\vset_0)$ of size
  $\setsize{\leftvertexsubset} = \expansionguarantee$
  for which it holds by expansion that
  $\setsize{\boundary^{\expandergraph}(\leftvertexsubset)}
  \geq
  2\expansionguarantee$.
  But this is a contradiction since
  as argued above we should have
  $\Setsize{\boundary^{\expandergraph}(\leftvertexsubset)}
  \leq 
  \setsize{\vset_0}
  \leq 
  \expansionguarantee/2$.

  For the induction step, suppose that
  \propclosureone  and   \propclosuredtwo both hold
  \mbox{for $i-1$}. 
  Let   us write
  $\ubar = \Ker(\vset_{i-1})\cup \uset_{i}$
  and consider any 
  $\leftvertexsubset$
  such that
  $
  \ubar
  \subseteq \leftvertexsubset \subseteq \Ker(\vset_i)$. 
  We claim that every 
  vertex
  in~$\boundary^{\expandergraph}(\utilde)$ is
  either   in the boundary
  $\boundary^{\expandergraph}(\ubar)$
  or is
  a member of~$\vset_0$.    To see why this is so, observe that since
  $\utilde\subseteq \Ker(\vset_i)$ we have
  $\boundary^{\expandergraph}(\utilde) 
  \subseteq \vset_i 
  = \vset_0 \cup
  \bigcup_{j=1}^{i} \nbhd^{\expandergraph}(\uset_j)$.  
  Furthermore,   note that
  $\bigcup_{j=1}^{i} \uset_j \subseteq
  \ubar\subseteq \utilde$
  holds (which is  due to the fact that
  $\nbhd(\Ker(\rightvertexsubset)) \subseteq \rightvertexsubset$
  for any~$\rightvertexsubset$). 
  Hence,
  for any
  $v \in \boundary^{\expandergraph}(\utilde) \setminus \vset_0$
  it must be the case that
  $v \in \bigcup_{j=1}^{i} \nbhd^{\expandergraph}(\uset_j)$,
  and so the unique neighbour of~$v$ on the left is
  contained in~$\bigcup_{j=1}^{i} \uset_j$ 
  and 
  therefore
  also in $\ubar$, implying that
  $v\in \boundary(\ubar)$. 
  From this we can conclude that
  \begin{equation}
    \label{eq:boundary-U-prime}
    \boundary^{\expandergraph}(\utilde)\setminus \vset_0 \subseteq
    \boundary^{\expandergraph}(\ubar)\setminus \vset_0 
    \eqcomma
  \end{equation}
  and we will use this to show that
  \begin{equation}
    \label{eq:boundary-U-star}
    \Setsize{\boundary^{\expandergraph}(\ubar) \setminus \vset_0}
    =
    \Setsize{\boundary^{\expandergraph}(\Ker(\vset_{i-1}) \cup
      \uset_{i}) \setminus \vset_0}
    \leq 
    \setsize{\Ker(\vset_i)}
  \end{equation}
  in order
  to prove \propclosureone.

  By definition, it holds that 
  every vertex in $\vset_{i-1}\setminus\vset_0$
  has at least one neighbour in $\Ker(\vset_{i-1})$. 
  It follows that for
  $\ubar = \Ker(\vset_{i-1})\cup \uset_{i}$
  all new boundary vertices in
  $
  \boundary^{\expandergraph}(\ubar)
  \setminus
  \boundary^{\expandergraph}(\Ker(\vset_{i-1}))
  $ 
  are either from~$\vset_0$  or from the boundary
  $\boundary^{\expandersubgraph{\expandergraph}{\vset_{i-1}}}(\uset_i)$
  of $\uset_i$ that lies outside of $\vset_{i-1}$.   
  Therefore we have
  \begin{equation}
    \label{eq:boundary-ker-V-iminus1-U-i}
    \boundary^{\expandergraph}(\ubar) \setminus \vset_0
    =
    \boundary^{\expandergraph} \bigl( \Ker(\vset_{i-1}) \cup
    \uset_{i} \bigr)
    \setminus \vset_0 
    \subseteq
    \bigl( \boundary^{\expandergraph}(\Ker(\vset_{i-1})) \setminus \vset_0\bigr)
    \disjointunion
    \boundary^{\expandersubgraph{\expandergraph}{\vset_{i-1}}}(\uset_i)    
    \eqperiod
  \end{equation}
  Since we have chosen $\uset_i$ so that it does not satisfy the
  expansion condition we know that
  \begin{equation}
    \label{eq:boundary-G-induced-U-i}
    \Setsize{\boundary^{\expandersubgraph{\expandergraph}{\vset_{i-1}}}(\uset_i)}
    \leq \setsize{\uset_i}    
  \end{equation}
  and by the inductive hypothesis for  \propclosureone it holds that
  \begin{equation}
    \label{eq:using-IH-prop-1}
    \Setsize{\boundary^{\expandergraph}(\Ker(\vset_{i-1})) \setminus \vset_0}
    \leq
    \setsize{\Ker(\vset_{i-1})} \eqperiod     
  \end{equation}
  Combining
  \refeq{eq:boundary-U-prime}
  with
  \mbox{\refeq{eq:boundary-ker-V-iminus1-U-i}--\refeq{eq:using-IH-prop-1}}
  we conclude that
  \begin{multline}
    \label{eq:induction-step-punchline-prop1}
    \Setsize{\boundary^{\expandergraph}(\leftvertexsubset)
      \setminus \vset_0}
    \leq
    \Setsize{\boundary^{\expandergraph}(\Ker(\vset_{i-1})
      \cup \uset_{i})
      \setminus \vset_0} 
    \leq 
    \\
    \leq 
    \Setsize{
      \bigl( \boundary^{\expandergraph}(\Ker(\vset_{i-1})) \setminus \vset_0\bigr)
    }
    +
    \Setsize{\boundary^{\expandersubgraph{\expandergraph}{\vset_{i-1}}}(\uset_i)}    
    \leq
    \setsize{\Ker(\vset_{i-1})} + \setsize{\uset_i}
    \leq
    \setsize{\Ker(\vset_i)}
    \eqcomma
  \end{multline}
  where the last inequality holds since
  $\Ker(\vset_{i-1})$ and $\uset_i$ are disjoint subsets of
  $\Ker(\vset_i)$.   
  This completes the inductive step for \propclosureone.
  
  To show \propclosuredtwo, let us first assume that
  $\setsize{\Ker(\vset_i)} \leq   \expansionguarantee$.  
  Then by the expansion properties of~$\expandergraph$
  together with
  \propclosureone applied to the set
  $\utilde = \Ker(\vset_i)$ 
  we have 
  \begin{equation}
    \label{eq:prop2-1}
    2\setsize{\Ker(\vset_i)}
    \leq
    \Setsize{\boundary^{\expandergraph}(\Ker(\vset_i))} 
    \leq 
    \setsize{\vset_0} + \setsize{\Ker(\vset_i)}    
    \eqcomma
  \end{equation}
  from which it follows that
  \begin{equation}
    \label{eq:prop2-2}
    \setsize{\Ker(\vset_i)}
    \leq
    \setsize{\vset_0}
    \eqperiod
  \end{equation} 
  If instead
  $\setsize{\Ker(\vset_i)} >
  \expansionguarantee$,  
  then by the inductive hypothesis we know that
  $
  \setsize{\Ker(\vset_{i-1})}
  \leq 
  \setsize{\vset_0}
  \leq 
  \expansionguarantee/2
  $ 
  and by
  construction we have
  $\setsize{\uset_i}\leq\expansionguarantee/2$. 
  Therefore,
  there must exist a vertex set~$\leftvertexsubset$ 
  of size~$\expansionguarantee$
  satisfying the condition
  $\Ker(\vset_{i-1})\cup\uset_i\subseteq\leftvertexsubset\subseteq
  \Ker(\vset_i)$
  in \propclosureone.
  From the expansion properties of~$\expandergraph$ we conclude that
  $\setsize{\boundary(\leftvertexsubset)}\geq2\expansionguarantee$,
  which is a contradiction because 
  for sets~$\leftvertexsubset$ 
  satisfying the conditions  in \propclosureone we
  derived~\refeq{eq:induction-step-punchline-prop1}, 
  which implies that
  $\setsize{\boundary(\leftvertexsubset)}\leq \setsize{\vset_0}
  +\setsize{\Ker(\vset_{i-1})} +\setsize{\uset_i} \leq
  3\expansionguarantee/2$. 
  The claim follows by the induction principle.
\end{proof}

\begin{lem:expanderexists}
  \newexpanderexistTEXT{}
\end{lem:expanderexists}

\begin{proof}
  Let $\leftvertexset$ and $\rightvertexset$ be two disjoint sets of
  vertices of size
  $
  \setsize{\leftvertexset}=\leftsize=
  \Floor{\rightsize^{\smalldelta\expanderdegree}}
  $
  and $\setsize{\rightvertexset}=\rightsize$.   
  For every $u\in\leftvertexset$ we choose $\expanderdegree$ times a
  neighbour $v\in\rightvertexset$ uniformly at random with
  repetitions.  
  This gives us a bipartite graph $\expandergraph = (\leftvertexset
  \disjointunion  \rightvertexset,E)$  of left-degree at most
  $\expanderdegree$.  
  In the sequel we show that $\expandergraph$ is almost surely an
  \boundaryexp{\expdegree}{\expguarantee}{2} as $\rightsize \to \infty$.  
  
  First note that for every set
  $\leftvertexsubset\subseteq\leftvertexset$ all neighbours $v\in
  \nbhd(\leftvertexsubset)\setminus\boundary(\leftvertexsubset)$ that
  are not in the boundary of $\leftvertexsubset$ have at least two
  neighbours in $\leftvertexsubset$.  
  Since there are at most
  $\expanderdegree\setsize\leftvertexsubset -
  \setsize{\boundary(\leftvertexsubset)}$ 
  edges between $\leftvertexsubset$ and
  $\nbhd(\leftvertexsubset)\setminus\boundary(\leftvertexsubset)$, it
  follows that
  $\setsize{\nbhd(\leftvertexsubset) \setminus
    \boundary(\leftvertexsubset)} 
  \leq
  (\expanderdegree\setsize\leftvertexsubset -
  \setsize{\boundary(\leftvertexsubset)})/2$ 
  and hence  
  \begin{multline}
    \label{eq:neighbourhoodbound}
    \setsize{\nbhd(\leftvertexsubset)} 
    =
    \Setsize{\nbhd(\leftvertexsubset) 
      \intersection
      \boundary(\leftvertexsubset)} 
    +
    \Setsize{\nbhd(\leftvertexsubset) \setminus \boundary(\leftvertexsubset)} 
    \leq
    \\
    \leq
    \setsize{\boundary(\leftvertexsubset)} +
    \frac{\expanderdegree\setsize\leftvertexsubset -
      \setsize{\boundary(\leftvertexsubset)}}
    {2} 
    =
    \frac{\expanderdegree\setsize\leftvertexsubset +
      \setsize{\boundary(\leftvertexsubset)}}
    {2} 
    \eqperiod
  \end{multline}
  If $\expandergraph$ is not an
  \boundaryexp{\expdegree}{\expguarantee}{2}, then there is a set
  $\leftvertexsubset$ of size
  $\leftvertexsubsetsize\leq\expansionguarantee$ that has a boundary
  $\boundary(\leftvertexsubset)$ of size at most $2
  \leftvertexsubsetsize$ and from \eqref{eq:neighbourhoodbound} it
  follows that $\setsize{\nbhd(\leftvertexsubset)}\leq
  (1+\expanderdegree/2)\leftvertexsubsetsize$.  
  By a union bound argument we obtain
  \begin{subequations}
    \begin{align}
      &\Pr[\expandergraph 
        \text{ is not an \boundaryexp{\expdegree}{\expguarantee}{2}}]
      \\
      \leq 
      & \sum^{\expansionguarantee}_{\leftvertexsubsetsize=1}      
        \sum_{
        \leftvertexsubset \subseteq [\leftsize] ; \, 
        \setsize{\leftvertexsubset} = \leftvertexsubsetsize
        }
        \Pr\big[\setsize{\boundary(\leftvertexsubset)}\leq 
        2 \leftvertexsubsetsize         
        \big] 
      \\
      \leq 
      &\sum^{\expansionguarantee}_{\leftvertexsubsetsize=1}
        \binom{\leftsize}{\leftvertexsubsetsize}
        \Pr\big[\setsize{\nbhd(\leftvertexsubset)}\leq
        (1+\expanderdegree/2)\leftvertexsubsetsize \text{ for some fixed }
        \setsize{\leftvertexsubset}=\leftvertexsubsetsize\big] 
      \\ 
      \leq 
      &\sum^{\expansionguarantee}_{\leftvertexsubsetsize=1}
        \binom{\leftsize}{\leftvertexsubsetsize}
        \binom{\rightsize}{(1+\expanderdegree/2)\leftvertexsubsetsize}
        \left(
        \frac{(1+\expanderdegree/2)\leftvertexsubsetsize}{\rightsize}
        \right)^{\expanderdegree\leftvertexsubsetsize} 
        \label{eq:binom}
      \\ 
      \leq 
      &\sum^{\expansionguarantee}_{\leftvertexsubsetsize=1} 
        \leftsize^\leftvertexsubsetsize 
        \left(
        \frac
        {\eulernumber\rightsize}
        {(1+\expanderdegree/2)\leftvertexsubsetsize}
        \right)^{(1+\expanderdegree/2)\leftvertexsubsetsize}
        \left(
        (1+\expanderdegree/2)\leftvertexsubsetsize
        \right)^{\expanderdegree\leftvertexsubsetsize}
        \rightsize^{-\expanderdegree\leftvertexsubsetsize} 
        \label{eq:binombound}
      \\
      = 
      &\sum^{\expansionguarantee}_{\leftvertexsubsetsize=1} 
        \leftsize^\leftvertexsubsetsize 
        (\eulernumber\rightsize)^{(1+\expanderdegree/2)\leftvertexsubsetsize}
        \left(
        (1+\expanderdegree/2)\leftvertexsubsetsize
        \right)^{(\expanderdegree/2-1)\leftvertexsubsetsize}
        \rightsize^{-\expanderdegree\leftvertexsubsetsize} 
      \\      
      \label{eq:fourth-from-end}
      \leq      
      &\sum^{\expansionguarantee}_{\leftvertexsubsetsize=1} 
        \rightsize^{\smalldelta\expanderdegree\ell}
        (\eulernumber\rightsize)^{(1+\expanderdegree/2)\leftvertexsubsetsize}
        \left(
        (1+\expanderdegree/2)\leftvertexsubsetsize
        \right)^{(\expanderdegree/2-1)\leftvertexsubsetsize}
        \rightsize^{-\expanderdegree\leftvertexsubsetsize} 
      \\
      \label{eq:third-from-end}
      = 
      & \sum_{\ell=1}^\expansionguarantee
        \rightsize^{\smalldelta\expanderdegree\ell}
        \rightsize^{
        \frac{\log \eulernumber}{\log \rightsize}
        (1+\expanderdegree/2)\leftvertexsubsetsize
        }
        \rightsize^{\frac{1}{\log \rightsize}
        \log\bigl(
        (\expanderdegree/2+1)\ell
        \bigr)
        (\expanderdegree/2-1)\leftvertexsubsetsize}
        \rightsize^{(-\expanderdegree/2+1)\ell}
    \\
      \label{eq:next-to-final}
      \leq  
      &\sum_{\ell=1}^\expansionguarantee
        \rightsize^{
        \bigl(
        \frac{\log \eulernumber}{\log \rightsize} \expanderdegree
        + \frac{1}{\log \rightsize}
        \log(\expanderdegree\expansionguarantee) 
        (\expanderdegree/2-1) 
        - \expanderdegree/2 + 1 +
        \smalldelta\expanderdegree
        \bigr)
        \ell} 
      \\
      \label{eq:final}
      =
      &\sum_{\ell=1}^\expansionguarantee
        \rightsize^{
        \bigl(
        \frac{\log \eulernumber}{\log \rightsize} +
        \frac{1}{\log \rightsize}
        \log(\expanderdegree\expansionguarantee) 
        (1/2-1/\expanderdegree)
        - 1/2 + 1/\expanderdegree +
        \smalldelta\bigr)\expanderdegree\ell}
        \eqcomma
    \end{align}
  \end{subequations}
  where to get
  from line \eqref{eq:binom} to \eqref{eq:binombound} we used that 
  $\binom{n}{k}\leq\left(\frac{\eulernumber{}n}{k}\right)^k$ for
  the Euler number $\eulernumber$,
  from~\refeq{eq:fourth-from-end} to~\refeq{eq:third-from-end}
  we used that 
  $n^{\log a / \log n} = a$,
  and from~\refeq{eq:third-from-end} to~\refeq{eq:next-to-final}
  that $\expanderdegree \geq \mindegree \geq 2$
  and 
  $\ell \leq \expansionguarantee$.
  In order to show that 
  \refeq{eq:final}
  is bounded away from 1, it suffices to 
  demonstrate
  that the   
  expression 
  \begin{equation}
    \frac{\log \eulernumber}{\log \rightsize}+\frac{1}{\log \rightsize}\log(
    \expanderdegree\expansionguarantee) 
    (1/2-1/\expanderdegree) 
    -1/2+1/\expanderdegree +
    \smalldelta
  \end{equation}
  is negative and bounded away from zero.
  Set
  $\lambda = \smallepsilon/2 - 1/\mindegree - \smalldelta>0$ 
  and choose $\rightsizezero = 3^{2/\lambda}$.
  By the upper bounds on $\expansionguarantee$ and $\expanderdegree$ it
  follows that 
  \begin{subequations}
    \begin{align}
      &\log \eulernumber/\log \rightsize+\log(
        \expanderdegree\expansionguarantee) 
        (1/2-1/\expanderdegree) /\log \rightsize
        -1/2+1/\expanderdegree +
        \smalldelta
      \\
      \leq 
      &\log \eulernumber/\log \rightsize+(1/2)\log(
        \rightsize^{\frac12-\smallepsilon}\rightsize^{\frac12}) /\log \rightsize
        -1/2+1/\expanderdegree +
        \smalldelta
      \\
      = 
      &\log \eulernumber/\log \rightsize-\smallepsilon/2 + 1/\expanderdegree +
        \smalldelta
      \\
      \leq
      &\log \eulernumber/\log \rightsize-\smallepsilon/2 + 1/\mindegree +
        \smalldelta
      \\
      =
      &\log \eulernumber/\log \rightsize-\lambda
      \\
      \leq 
      &-\lambda/2
        \eqcomma
    \end{align}
  \end{subequations}
  where the last inequality holds since
  $\rightsize\geq\rightsizezero > \eulernumber^{2/\lambda}$.
  It follows that the
  probability that $\expandergraph$ is not an
  \boundaryexp{\expdegree}{\expguarantee}{2} is bounded by
  \begin{align}
    \sum_{\ell=1}^\expansionguarantee
    \rightsize^{\bigl(-\lambda/2\bigr)\expanderdegree\ell} \leq
    \sum_{\ell=1}^\expansionguarantee
    \rightsizezero^{\bigl(-\lambda/2\bigr)\expanderdegree\ell} \leq
    \sum_{\ell=1}^\infty
    \bigl(\tfrac13\bigr)^{\expanderdegree\ell} \leq \tfrac12
    \eqcomma
  \end{align}
  which establishes the lemma.
\end{proof}

\end{document}